\documentclass[envcountsect,envcountsame,oribibl]{llncs}
\usepackage{amssymb}
\usepackage[]{graphicx}
\usepackage[]{amsmath}
\usepackage{gn-logic14}
\usepackage{comment}
\usepackage[asymmetric,left=1in,top=1in,right=1in,bottom=0.5in]{geometry}

\title{Does Treewidth Help in Modal Satisfiability? (Extended Abstract)}
\author{M. Praveen}
\institute{The Institute of Mathematical Sciences, Chennai, India}

\begin{document}
\pagestyle{plain}
%Natural numbers
\newcommand{\macNat}{\mathbb{N}}

%big O notation
\newcommand{\Oh}{\mathcal{O}}

%treewidth and pathwidth
\newcommand{\mactw}{\text{tw}}
\newcommand{\macpw}{p{\scriptstyle w}}
\newcommand{\macheight}{h}
\newcommand{\mactwo}{t{\scriptstyle w}}

%complexity classes
\newcommand{\maccc}[1]{\textsc{#1}}
\newcommand{\macpspace}{\maccc{Pspace}}
\newcommand{\macfpt}{\maccc{Fpt}}
\newcommand{\macwone}{\maccc{W[1]}}
\newcommand{\macwtwo}{\maccc{W[2]}}
\newcommand{\macxp}{\maccc{Xp}}
\newcommand{\macnp}{\maccc{Np}}
\newcommand{\macp}{\maccc{Ptime}}
\newcommand{\macyes}{\maccc{Yes}}
\newcommand{\macoh}{\mathcal{O}}

\newcommand{\macppwcnfsat}{p-\textsc{Pw-Sat} }
\newcommand{\macnlcp}{\textsc{Nlcp}}
\newcommand{\maccsp}{\textsc{Csp}}

\newcommand{\macpopterm}[1]{\textit{#1}}

%propositional variables
\newcommand{\macpv}{\Phi}
\newcommand{\macpvo}{q}
\newcommand{\macpvt}{r}
\newcommand{\macpvh}{s}
\newcommand{\macpvf}{t}
\newcommand{\macpvidx}{i}

%modal logic formulas
\newcommand{\macmfo}{\phi}
\newcommand{\macmft}{\psi}
\newcommand{\macmd}{\text{md}}
\newcommand{\macpcnff}{\mathcal{F}}
\newcommand{\macpart}{part}
\newcommand{\mactget}{tg}
\newcommand{\macpartidx}{p}
\newcommand{\macnumpv}{n}
%models
\newcommand{\macmodelo}{\mathcal{M}}
\newcommand{\macworlds}{W}
\newcommand{\macworldo}{w}
\newcommand{\macworldt}{v}
\newcommand{\macfro}{\mathcal{A}}
\newcommand{\macrelo}{\mapsto}
\newcommand{\macvalo}{V{\scriptstyle l}}

%problem instances
\newcommand{\macprino}{I}
\newcommand{\macpro}{\Pi}
\newcommand{\macparamo}{\kappa}
\newcommand{\macalpho}{\Sigma}
\newcommand{\macparamdepo}{f}
\newcommand{\macparamdept}{g}
\newcommand{\macredmap}{A}
\newcommand{\macnuminst}{m}
\newcommand{\macinstidx}{i}

%relational structures
\newcommand{\macrelvoc}{\tau}
\newcommand{\macrelt}{R}
\newcommand{\macdom}{D}
\newcommand{\macdomel}{e}
\newcommand{\macarity}{\mathrm{\mathit{arity}}}
\newcommand{\macstruc}{\mathcal{S}}
\newcommand{\macedge}{O{\scriptstyle c}}
\newcommand{\macedgec}{\overline{\macedge}}
\newcommand{\mactree}{\mathcal{T}}
\newcommand{\mactn}{t}
\newcommand{\macbag}{B}
\newcommand{\maccls}{C{\scriptstyle l}}
\newcommand{\macclsc}{L{\scriptstyle t}}
\newcommand{\maclev}{L{\scriptstyle v}}
\newcommand{\macbox}{B{\scriptscriptstyle \Box}}
\newcommand{\macdmd}{D\diamond}
\newcommand{\macuc}{U}

%MSO formulas
\newcommand{\macmsofo}{\chi}
\newcommand{\macmsoft}{\xi}
\newcommand{\macmsofh}{\zeta}
\newcommand{\macmsoff}{\chi}
\newcommand{\macfovo}{x}
\newcommand{\macfovt}{y}
\newcommand{\macfovh}{z}
\newcommand{\macsovo}{X}
\newcommand{\macsovt}{Y}

%CNF for modal logic
\newcommand{\maclit}{\rm{\mathit{literal}}}
\newcommand{\maccl}{\rm{\mathit{clause}}}
\newcommand{\maccnf}{\rm{\mathit{CNF}}}

%Indices
\newcommand{\macmdidx}{i}
\newcommand{\macnumcl}{m}

%For referring to definitions, theorems etc.
\newcommand{\Defref}[1]{Definition~\ref{#1}}
\newcommand{\defref}[1]{Def.~\ref{#1}}
\newcommand{\Thmref}[1]{Theorem~\ref{#1}}
\newcommand{\thmref}[1]{Theorem~\ref{#1}}
\newcommand{\Lemref}[1]{Lemma~\ref{#1}}
\newcommand{\lemref}[1]{Lemma~\ref{#1}}
\newcommand{\Propref}[1]{Proposition~\ref{#1}}
\newcommand{\propref}[1]{Prop.~\ref{#1}}
\newcommand{\Claimref}[1]{Claim~\ref{#1}}
\newcommand{\claimref}[1]{Claim~\ref{#1}}
\newcommand{\algoref}[1]{Algorithm~\ref{#1}}
\newcommand{\Algoref}[1]{Algorithm~\ref{#1}}
\newcommand{\Figref}[1]{Figure~\ref{#1}}
\newcommand{\figref}[1]{Fig.~\ref{#1}}
\newcommand{\Tabref}[1]{Table~\ref{#1}}
\newcommand{\tabref}[1]{Table~\ref{#1}}
\newcommand{\Secref}[1]{Section~\ref{#1}}
\newcommand{\secref}[1]{section~\ref{#1}}
\newcommand{\Chref}[1]{Chapter~\ref{#1}}
\newcommand{\chref}[1]{Chap.~\ref{#1}}
\newcommand{\apndref}[1]{Appendix~\ref{#1}}
\newcommand{\Apndref}[1]{App.~\ref{#1}}

%Expressing modal satisfiability in MSO
\newcommand{\macclsu}{C{\scriptscriptstyle \ell}}
\newcommand{\mactlitsu}{T{\scriptstyle r}}
\newcommand{\macclel}{c{\scriptscriptstyle \ell}}
\newcommand{\maclitel}{l{\scriptstyle t}}
\newcommand{\maclevsat}[1]{\macmsoft[#1]}
\newcommand{\macmdsat}[1]{\macmsofh[#1]}
\newcommand{\macreach}{\mathcal{R}}
\newcommand{\macmdglsat}{\macmsoff}
\newcommand{\macdef}{\overset{\bigtriangleup}{=}}
\newcommand{\maccomits}{C{\scriptstyle m}}
\newcommand{\macdemands}{D{\scriptstyle m}}
\newcommand{\macpvsul}{PV}
\newcommand{\macbfsul}{BL}
\newcommand{\macdfsul}{DL}
\newcommand{\macpvs}{P{\scriptstyle v}}
\newcommand{\machle}{H{\scriptscriptstyle l}}
\newcommand{\macgllitsu}{G{\scriptscriptstyle \ell t}}
\newcommand{\macglclsu}{GC{\scriptscriptstyle \ell}}
\newcommand{\macclidx}{i}
\newcommand{\maclitidx}{i}

%Constants
\newcommand{\macico}{c}

%Graphs
\newcommand{\macgraph}{G}
\newcommand{\macvertexo}{v}
\newcommand{\macvertext}{u}
\newcommand{\macuedge}{e}
\newcommand{\macvertexs}{V}
\newcommand{\macedges}{E}
\newcommand{\maccols}{S}
\newcommand{\maccolor}{\ell}
\newcommand{\macatleast}{atLeast}
\newcommand{\macatmost}{atMost}

% target counting model
\newcommand{\macpartindo}{t{\scriptstyle \uparrow}}
\newcommand{\macpartindt}{f{\scriptstyle \uparrow}}
\newcommand{\macpartctro}{t{\scriptstyle r}}
\newcommand{\macpartctrt}{f{\scriptstyle l}}
\newcommand{\macdpthctro}{d}
\newcommand{\macproper}{proper}
\newcommand{\macdetermined}{determined}
\newcommand{\maccountinit}{countInit}
\newcommand{\macdepth}{depth}
\newcommand{\maccountmonotone}{countMonotone}
\newcommand{\macsetcounter}{setCounter}
\newcommand{\macinccounter}{incCounter}
\newcommand{\mactargetmet}{targetMet}
\newcommand{\macnumpart}{k}
\newcommand{\mactcidx}{j}
\newcommand{\macsatasgn}{f}

\newcommand{\macstar}{[*]}

\maketitle

\includecomment{details}
\excludecomment{wellknown}
\excludecomment{conf}
\includecomment{minorext}

\begin{abstract}
  Many tractable algorithms for solving the Constraint Satisfaction
  Problem (\maccsp) have been developed using the notion of the
  treewidth of some graph derived from the input \maccsp\/ instance.
  In particular, the \macpopterm{incidence graph} of the
  \maccsp{} instance is one such graph. We introduce the notion of an
  incidence graph for modal logic formulae in a certain normal form.
  We investigate the parameterized complexity of modal satisfiability
  with the modal depth of the formula and the treewidth of the
  incidence graph as parameters.  For various combinations of
  Euclidean, reflexive, symmetric and transitive models, we show
  either that modal satisfiability is \macfpt, or that it is
  \macwone-hard.  In particular, modal satisfiability in general
  models is \macfpt{}, while it is \macwone-hard in transitive models.
  As might be expected, modal satisfiability in transitive and
  Euclidean models is \macfpt.
\end{abstract}

\section{Introduction}
Treewidth as a parameter has been very successful in obtaining
Fixed Parameter Tractable (\macfpt{}) algorithms for many classically
intractable problems. One such class of problems is constraint
satisfaction and closely related problems like satisfiability in
propositional logic and the homomorphism problem \cite{DKV02,SS10}.
There have been recent extensions to quantified constraint
satisfaction \cite{C04,PV06}. In such problems, treewidth is used as a
measure of modularity inherent in the given problem instance and
algorithms make use of the modularity to increase their efficiency.
Understanding the extent to which treewidth can be stretched in such
problems is an active area of research \cite{M07,G06}. This work
explores the applicability of such techniques to modal satisfiability.

Apart from having many applications (reasoning about knowledge
\cite{FHMV95}, programming \cite{P80} and hardware verification
\cite{RS83} etc.), modal logics have nice computational properties
\cite{Vardi96,Gradel01}. Many tools have been built for
checking satisfiability of modal formulae \cite{HS99,PV2003}, despite
being intractable in the classical sense (\macpspace-complete or
\macnp-complete in most cases). Complexity of modal logic decision problems
is well studied \cite{L77,HM92,H95}. Another motivation for this work
is to strengthen the complexity classification of modal logics through
the refined analysis offered by parameterized complexity.

\emph{Our results:} It is known that any modal logic formula can be
effectively converted into a Conjunctive Normal Form (CNF)
\cite{EC89,HM08}. Given a modal logic formula in CNF, we associate a
graph with it. Restricted to propositional CNF formulae (which are
modal formulae with modal depth $0$), this graph is precisely the
\macpopterm{incidence graph} associated with propositional CNF
formulae (see \cite{SS10} for details). We prove that
\begin{enumerate}
  \item with the treewidth of the graph and the modal depth of the
    formula as parameters, satisfiability in general models is
    \macfpt,
  \item with treewidth and modal depth as parameters, satisfiability
    in transitive models is \macwone-hard and
  \item with treewidth as the parameter, satisfiability in models that
    are Euclidean\footnote{A binary relation $\macrelo$ is Euclidean
    if $\forall \macfovo,\macfovt,\macfovh$, $\macfovo\macrelo
    \macfovt$ and $\macfovo\macrelo \macfovh$ implies
    $\macfovt\macrelo \macfovh$.} and any combination of reflexive,
    symmetric and transitive is \macfpt\/.
\end{enumerate}
Since modal formulae of modal depth $0$ contain all propositional
formulae, bounding modal depth alone will not give \macfpt\/ results
(unless \macp=\macnp).  The main idea behind our \macfpt\/ results is
to express satisfiability of a modal formula in Monadic Second Order
(MSO) logic over the formula's associated graph and then apply
Courcelle's theorem \cite{C92}. Modal formulae with low treewidth are
quite powerful, capable of encoding complex problems (see the
conclusion for relevant pointers). On the other hand, modal formulae
with low treewidth contain propositional CNF formulae of low
treewidth, which arise naturally in many practical applications. See
\cite[Section 1.4]{FMR08} and references therein for some context on
this.

\emph{Related work:} In \cite{H95}, Halpern considers the effect of
bounding different parameters (such as the number of propositional
variables, modal depth etc., but not treewidth) on complexity. In
\cite{LAN05}, Nguyen shows that satisfiability of many modal logics
reduce to \macp\/ under the restriction of Horn fragment and bounded
modal depth. In \cite{ALM09}, Achilleos et.~al.~consider parameterized
complexity of modal satisfiability in general models with the number of
propositional variables and other structural aspects (but not
treewidth) as parameters. In \cite{AW09}, Adler et.~al.~associate
treewidth with First Order (FO) formulae and use it to obtain a \macfpt\/
algorithm for model checking.

The Complexity of satisfiability of modal logics follow a pattern. In
\cite{HR07}, Halpern et.~al.~prove that with the addition of Euclidean
property, complexity of (infinitely) many modal logics drop from
\macpspace-hard to \macnp-complete. \cite{HS08} is another work in
this direction. Similar pattern is observed in graded modal logics
\cite{KP09}. With treewidth and modal depth as parameters, our results
indicate similar behaviour in the world of parameterized complexity
--- satisfiability in transitive models is \macwone-hard, while
satisfiability in Euclidean and transitive models is \macfpt\/, even
with treewidth as the only parameter.  However, more work is needed in
this direction. First, the results in \cite{HR07,HS08} hold for
infinitely many cases while we consider only a few fixed cases.
Second, satisfiability in general models is \macpspace-complete and
drops to \macnp-complete with the addition of Euclidean property. In
our setting, satisfiability in general models is already \macfpt\/
(but see conclusion for a discussion about why satisfiability in
general models is not \macfpt\/ unless \macp=\macnp, when treewidth is
the only parameter).

\section{Preliminaries}
Let $\macNat$ denote the set of natural numbers. For $\macnumpart\in
\macNat$, we denote the set $\{1,\dots,\macnumpart\}$ by
$[\macnumpart]$.  We use standard notation about parameterized
complexity like \macfpt\/ algorithms, \macfpt\/ reductions and
\macwone-hardness from \cite{FG06}. We will also use notation and
definitions of relational structures and their tree decompositions
from \cite{FG06}: a \macpopterm{relational vocabulary} $\macrelvoc$ is
a set of relation symbols. Each relation symbol $\macrelt$ has an
arity $\macarity(\macrelt)\ge 1$. A $\macrelvoc$-structure $\macstruc$
consists of a set $\macdom$ called the \macpopterm{domain} and an
interpretation $\macrelt^{\macstruc}\subseteq
\macdom^{\macarity(\macrelt)}$ of each relation symbol $\macrelt\in
\macrelvoc$. A \macpopterm{graph} is an $\{\macedges\}$-structure,
where $\macedges$ is a binary edge relation.  A \macpopterm{tree} is a
graph without cycles. A \macpopterm{path decomposition} is a tree
decomposition \cite[Definition 11.23]{FG06} whose underlying tree is a
path.  The pathwidth of a structure is the minimum of the widths of
all path decompositions.  It is known that computing optimal tree and
path decompositions of a relational structure is \macfpt\/ when
parameterized by treewidth; cf. \cite[Corollary 11.28]{FG06} and
\cite{BK96}.

Courcelle's theorem (\cite[Theorem 11.37]{FG06}) states that given a
relational structure and a MSO sentence, checking whether the MSO
sentence is true in the structure is \macfpt\/ when parameterized by
the treewidth of the structure and the length of the sentence.
\begin{wellknown}
  \subsection{Parameterized complexity}
  We assume that problem instances are given as strings in some alphabet
  $\macalpho$. A parameterized problem is of the
  form $(\macpro,\macparamo)$, where $\macpro\subseteq
  \macalpho^{*}$ is a set of \macyes\/ instances of the problem and
  $\macparamo:\macalpho^{*}\to \macNat$ is a parameter. A
  parameterized problem is said to be \macpopterm{Fixed Parameter
  Tractable} (\macfpt) if there is an algorithm that takes any instance
  $\macprino\in \macalpho^{*}$ of the problem $\macpro$ and decides
  within time
  $\macparamdepo(\macparamo(\macprino))|\macprino|^{\Oh(1)}$ whether
  $\macprino\in \macpro$. Here, $\macparamdepo$ is some computable
  function and $|\macprino|$ is the length of $\macprino$. The problem
  of computing $\macparamo(\macprino)$ must itself be in polynomial time
  or \macfpt\/ with $\macparamo$ itself as the parameter.

  If $(\macpro,\macparamo)$ and $(\macpro',\macparamo')$ are
  parameterized problems over alphabets $\macalpho$ and $\macalpho'$
  respectively, $(\macpro,\macparamo)$ is \macfpt\/ reducible to
  $(\macpro',\macparamo')$ if there is a mapping
  $\macredmap:\macalpho^{*}\to (\macalpho')^{*}$ such that:
  \begin{enumerate}
    \item For all $\macprino\in \macalpho^{*}$, $\macprino\in \macpro$
      iff $\macredmap(\macprino)\in \macpro'$.
    \item There is an algorithm that terminates in time at most
      $\macparamdepo(\macparamo(\macprino))|\macprino|^{\Oh(1)}$ and
      computes $\macredmap(\macprino)$, where $\macparamdepo$ is some
      computable function.
    \item There is a computable function $\macparamdept:\macNat\to\macNat$
      such that $\macparamo'(\macredmap(\macprino))\le
      \macparamdept(\macparamo(\macprino))$.
  \end{enumerate}

  Similar to hierarchy of intractable classes in classical complexity
  theory, there is a hierarchy of intractable classes in parameterized
  complexity also. A few of these classes are $\macfpt\subseteq
  \macwone\subseteq \macwtwo \subseteq \dots\subseteq \macxp$. While
  inclusion between \macfpt\/ and \macxp\/ is known to be strict, strictness
  of other inclusions are not known. It is widely believed that all
  inclusions are strict. A parameterized problem $(\macpro,\macparamo)$
  is said to be \macwone-hard if all parameterized problems in
  \macwone\/
  are \macfpt\/ reducible to $(\macpro,\macparamo)$. Hence, \macwone-hard
  problems are unlikely to be \macfpt. Suppose a
  parameterized problem $(\macpro,\macparamo)$ is \macfpt\/ reducible to
  $(\macpro',\macparamo')$. If $(\macpro',\macparamo')$ is \macfpt, then
  $(\macpro,\macparamo)$ is also \macfpt. On the other hand, if
  $(\macpro,\macparamo)$ is \macwone-hard, then $(\macpro',\macparamo')$
  is also \macwone-hard.

  \subsection{Relational structures, treewidth and MSO logic}
  We will follow the notation of \cite{FG06}. A \macpopterm{relational
  vocabulary} $\macrelvoc$ is a set of relation symbols. Each relation
  symbol $\macrelt$ has an \macpopterm{arity} $\macarity(\macrelt)\ge
  1$. A $\macrelvoc$-structure $\macstruc$ consists of a set $\macdom$
  called the \macpopterm{domain} and an interpretation
  $\macrelt^{\macstruc}\subseteq \macdom^{\macarity(\macrelt)}$ of each
  relation symbol $\macrelt\in \macrelvoc$. A \macpopterm{graph} is an
  $\{\macedges\}$-structure, where $\macedges$ is a binary edge relation.
  A \macpopterm{tree} is a graph without cycles.

  A \macpopterm{tree decomposition} of a $\macrelvoc$-structure
  $\macstruc$ is a pair $(\mactree,(\macbag_{\mactn})_{\mactn\in
  \mactree})$, where $\mactree$ is a tree and
  $(\macbag_{\mactn})_{\mactn\in \mactree}$ is a family of subsets of
  the domain $\macdom$ of $\macstruc$ such that:
  \begin{enumerate}
    \item For all $\macdomel\in \macdom$, the set $\{\mactn\in
      \mactree\mid \macdomel\in \macbag_{\mactn}\}$ is nonempty and
      connected in $\mactree$.
    \item For every relation symbol $\macrelt\in \macrelvoc$ and every
      tuple $(\macdomel_{1},\ldots,\macdomel_{\macarity(\macrelt)})\in
      \macrelt^{\macstruc}$, there is a $\mactn\in \mactree$ such that
      $\macdomel_{1},\ldots,\macdomel_{\macarity(\macrelt)} \in
      \macbag_{\mactn}$.
  \end{enumerate}
  The width of such a decomposition is the number
  $\max\{|\macbag_{\mactn}|\mid \mactn\in \mactree\}-1$. The
  \macpopterm{treewidth} $\mactw(\macstruc)$ of $\macstruc$ is the
  minimum of the widths of tree decompositions of $\macstruc$. By
  replacing tree by path everywhere in the above definitions, we get
  \macpopterm{path decomposition} and \macpopterm{pathwidth}.

  For algorithmic purposes, we assume that relational
  structures are encoded as explained in \cite[Section 4.2]{FG06}. It is
  known that computing an optimal tree decomposition of a relational
  structure is \macfpt\/ when parameterized by treewidth \cite[Corollary
  11.28]{FG06}. Since an optimal path decomposition can be computed in
  \macfpt\/ if input is given along with an optimal tree decomposition
  \cite{BK96} and since treewidth is less than or equal to pathwidth,
  computing an optimal path decomposition of a relational structure is
  also \macfpt\/ when parameterized by pathwidth.

  We assume familiarity with Monadic Second Order (MSO) logic of
  relational structures and Courcelle's theorem (cf. \cite[Section 4.2,
  Theorem 11.37]{FG06}). Courcelle's theorem states that given a
  relational structure and a MSO sentence, checking whether the
  structure satisfies the MSO sentence is \macfpt\/ when parameterized by
  treewidth of the structure and length of the sentence. We denote first
  order variables by $\macfovo,\macfovt,\ldots$, second order variables
  by $\macsovo,\macsovt,\ldots$ and MSO formulae by
  $\macmsofo,\macmsoft,\ldots$. In some cases, we will use more
  meaningful symbols like $\macclel$ for a first order variable
  intended to represent a clause and $\macclsc$ for a second order
  variable intended to represent a set of literals.

  \subsection{Modal logic}
  The basic modal language is defined using a set of propositional
  variables $\macpv$ (whose elements are usually denoted
  $\macpvo,\macpvt,\macpvh$ and so on) and unary modal operators
  $\Diamond$ (`diamond') and $\Box$ (`box'). The well-formed formulae
  $\macmfo$ of the basic modal language are given by the rule
  \begin{align*}
    \macmfo ::= \macpvo~|~\bot~|~\lnot\macmfo~|~\macmfo\lor\macmft~
    |~\Diamond\macmfo~|~\Box\macmfo
  \end{align*}
  where $\macpvo$ ranges over $\macpv$. This means that a formula is
  either a propositional variable, the propositional constant false
  (`bottom'), a negated formula, a disjunction of formulae or a formula
  prefixed by a diamond or a box.  We use standard abbreviations like
  $\macmfo\land \macmft \equiv \lnot (\lnot\macmfo \lor \lnot\macmft)$,
  $\macmfo\Rightarrow \macmft \equiv \lnot\macmfo \lor \macmft$ etc.

  A \macpopterm{Kripke model} for the basic modal language is a triple
  $\macmodelo=(\macworlds,\macrelo,\macvalo)$, where $\macworlds$ is a
  set of worlds, $\macrelo$ is a binary \macpopterm{accessibility}
  relation on $\macworlds$ and $\macvalo:\macworlds\times \macpv\to
  \{\top,\bot\}$ is the valuation function. If
  $\macvalo(\macworldo,\macpvo)=\top$, we think of it as $\macpvo$ being
  true in the world $\macworldo$. For $\macworldo,\macworldt\in
  \macworlds$, if $\macworldo\macrelo\macworldt$, $\macworldt$ is said
  to be a \macpopterm{successor} of $\macworldo$. The pair
  $(\macworlds,\macrelo)$ is called the \macpopterm{frame} $\macfro$
  underlying $\macmodelo$. If $\macrelo$ is reflexive, then $\macfro$
  and $\macmodelo$ are said to be a reflexive frame and a reflexive
  model respectively. Similar nomenclature is followed for other
  properties of $\macrelo$. The relation $\macrelo$ is Euclidean if for
  all $\macworldo_{1},\macworldo_{2},\macworldo_{3}$,
  $\macworldo_{1}\macrelo \macworldo_{2}$ and $\macworldo_{1} \macrelo
  \macworldo_{3}$ implies $\macworldo_{2}\macrelo \macworldo_{3}$.

  We denote the fact that a modal formula $\macmfo$ is satisfied at a
  world $\macworldo$ in a model $\macmodelo$ by $\macmodelo,\macworldo
  \models \macmfo$. The $\models$ relation is defined inductively as
  follows.
  \begin{align*}
    \macmodelo,\macworldo &\models \macpvo \text{ iff }
    \macvalo(\macworldo,\macpvo)=\top \\
    \macmodelo,\macworldo &\models \bot \text{ never}\\
    \macmodelo,\macworldo &\models \lnot \macmfo \text{ iff not }
    \macmodelo,\macworldo \models\macmfo \\
    \macmodelo,\macworldo &\models \macmfo\lor\macmft \text{ iff }
    \macmodelo,\macworldo \models \macmfo \text{ or }
    \macmodelo,\macworldo \models \macmft \\
    \macmodelo,\macworldo &\models \Diamond\macmfo \text{ iff for some
    successor }
    \macworldt \text{ of }
    \macworldo \text{, } \macmodelo,\macworldt
    \models \macmfo \\
    \macmodelo,\macworldo &\models \Box\macmfo \text{ iff for all
    successors }
    \macworldt \text{ of }
    \macworldo \text{, } \macmodelo,\macworldt
    \models \macmfo
  \end{align*}
\end{wellknown}

We use standard notation for modal logic from \cite{BDV01}: well
formed modal logic formulae are defined by the grammar $\macmfo ::=
\macpvo\in \macpv~|~\bot~|~\lnot\macmfo~|~\macmfo\lor~\macmft~
|~\Diamond\macmfo~|~\Box\macmfo$, where $\macpv$ is a set of
propositional variables. A \macpopterm{Kripke model} for the basic
modal language is a triple
$\macmodelo=(\macworlds,\macrelo,\macvalo)$, where $\macworlds$ is a
set of worlds, $\macrelo$ is a binary \macpopterm{accessibility}
relation on $\macworlds$ and $\macvalo:\macworlds\times \macpv\to
\{\top,\bot\}$ is a valuation function. For $\macworldo,\macworldt\in
\macworlds$, if $\macworldo\macrelo\macworldt$, $\macworldt$ is said
to be a \macpopterm{successor} of $\macworldo$. The pair
$(\macworlds,\macrelo)$ is called the \macpopterm{frame} $\macfro$
underlying $\macmodelo$. If $\macrelo$ is reflexive, then $\macfro$
and $\macmodelo$ are said to be a reflexive frame and a reflexive
model respectively. Similar nomenclature is followed for other
properties of $\macrelo$. The relation $\macrelo$ is Euclidean if for
all $\macworldo_{1},\macworldo_{2},\macworldo_{3}$,
$\macworldo_{1}\macrelo \macworldo_{2}$ and $\macworldo_{1} \macrelo
\macworldo_{3}$ implies $\macworldo_{2}\macrelo \macworldo_{3}$. We
denote the fact that a modal formula $\macmfo$ is satisfied at a world
$\macworldo$ in a model $\macmodelo$ by $\macmodelo,\macworldo \models
\macmfo$. For $\macpvo\in \macpv$, $\macmodelo,\macworldo\models
\macpvo$ iff $\macvalo(\macworldo,\macpvo)=\top$. Negation $\lnot$
and disjunction $\lor$ are treated in the standard way. For any
formula $\macmfo$, $\macmodelo,\macworldo\models \Diamond\macmfo$
($\macmodelo,\macworldo\models \Box \macmfo$) iff some (all)
successor(s) $\macworldt$ of $\macworldo$ satisfy
$\macmodelo,\macworldt\models \macmfo$. A modal formula
$\macmfo$ is \macpopterm{satisfiable} if there is a model $\macmodelo$
and a world $\macworldo$ in $\macmodelo$ such that
$\macmodelo,\macworldo \models \macmfo$.
\begin{details}
  A world $\macworldo'$ is said to be \macpopterm{reachable} from
  $\macworldo$ if there are worlds
  $\macworldo_{1},\macworldo_{2},\dots,\macworldo_{\macnuminst}$ such
  that $\macworldo\macrelo\macworldo_{1}\macrelo\cdots\macrelo
  \macworldo_{\macnuminst}\macrelo\macworldo'$. It is well known that
  if some modal formula is satisfied at some world $\macworldo$ in
  some Kripke model, discarding worlds not reachable from $\macworldo$
  does not affect satisfiability \cite[Proposition 2.6]{BDV01}.
  Henceforth, if some modal formula is satisfied at some world
  $\macworldo$ in some Kripke model $\macmodelo$, we will assume that
  $\macmodelo$ consists of only those worlds reachable from
  $\macworldo$.
\end{details}
Satisfiability in general,
reflexive and transitive models are all \macpspace-complete
\cite{L77}, while in equivalence models, it is \macnp-complete
\cite{L77}.

The modal depth $\macmd(\macmfo)$ of a modal formula $\macmfo$ is
inductively defined as follows. $\macmd(\macpvo)=\macmd(\bot)=0$.
$\macmd(\lnot\macmfo)=\macmd(\macmfo)$.
$\macmd(\macmfo\lor\macmft)=\max \{\macmd(\macmfo),\macmd(\macmft)\}$.
$\macmd(\Diamond\macmfo)=\macmd(\Box\macmfo)=\macmd(\macmfo)+1$.
We will use the Conjunctive Normal Form (CNF) for modal logic defined
in \cite{HM08}:
\begin{align*}
  \maclit &::= \macpvo~|~\lnot\macpvo~|~\Box\maccl~|~\Diamond\maccnf\\
  \maccl &::= \maclit~|~\maccl\lor\maccl~|~\bot\\
  \maccnf &::= \maccl~|~\maccnf\land\maccnf
\end{align*}
where $\macpvo$ ranges over $\macpv$. Any arbitrary modal formula
$\macmfo$ can be effectively transformed into CNF preserving
satisfiability \cite{EC89}. A $\maccnf$ is a conjunction of clauses and
a $\maccl$ is a disjunction of literals. A $\maclit$ is either a
propositional variable, a negated propositional variable or a formula of
the form $\Box\maccl$ or $\Diamond\maccnf$.  If one of the many
literals in a clause is $\bot$, then $\bot$ can be ignored without
affecting satisfiability. A literal of the form $\Diamond\bot$ can
similarly be ignored. However, a clause that has $\bot$ as the only
literal cannot be ignored since $\Box\bot$ is satisfied by a world in
some Kripke model iff that world has no successors. Henceforth, we
will assume that $\bot$ occurs only inside sub-formulae of the form
$\Box\bot$.

Suppose $\macmfo$ is a modal formula in CNF. If $\macmfo$ is of the
form $\maccl_{1}\land \maccl_{2}\land\cdots\land \maccl_{\macnumcl}$,
then $\maccl_{1},\maccl_{2},\ldots,\maccl_{\macnumcl}$ and all
literals appearing in these clauses are said to be at
\macpopterm{level} $\macmd(\macmfo)$. If $\Box\maccl_{1}$ is a
$\maclit$ at some level $\macmdidx$, then $\maccl_{1}$ and all
literals occurring in $\maccl_{1}$ are said to be
at level $\macmdidx-1$. If $\Diamond\maccnf$ is a literal at some
level $\macmdidx$ and $\maccnf$ is of the form
$\maccl_{1}\land\maccl_{2} \land \cdots \land \maccl_{\macnumcl'}$,
then $\maccl_{1},\maccl_{2}, \cdots, \maccl_{\macnumcl'}$ and all
literals appearing in these clauses are said to be at level
$\macmdidx-1$.  Note that a single propositional variable can occur in
the form of a $\maclit$ at different levels. The concept of level is
similar to the concept of distance defined in \cite{PV2003}. The
process of checking satisfiability we describe in
\secref{sec:modSatGenFrames} can be considered a variant of the
level-based bottom-up algorithm given in \cite{PV2003}, which is also
implicitly used in \cite[Theorem 5]{ALM09}. It requires more work and
combination of other ideas to prove that this process can be
formalized in MSO logic.

\begin{conf}
Proofs of lemmata marked with (*) are skipped due to
lack of space. A full version of this paper with the same title is
available at arXiv, which contains all the proofs.
\end{conf}

\section{Modal satisfiability in general models}
\label{sec:modSatGenFrames}
In this section, we will associate a relational structure with a modal
CNF formula. We show that checking satisfiability of a modal CNF
formula is \macfpt, parameterized by modal depth and the treewidth of
the associated relational structure. We begin with an example modal
CNF formula.

Consider the modal CNF formula $\left\{ \lnot\macpvo \lor \Box \left[
\macpvt\lor\lnot\macpvh \right] \right\}\land \left\{ \macpvo\lor
\Diamond\bot\right\}\land \left\{ \macpvt\lor \Diamond\left[
\lnot\macpvh\right] \right\}\land \left\{ \lnot\macpvt \lor
\Diamond\left[ (\macpvf\lor \lnot\macpvh) \land (\macpvt)
\right]\right\}$. Its modal depth is $1$ and has $4$ clauses at level
$1$. \Figref{fig:exampleMFRelStruc} shows a graphical representation
of this formula, which is very similar to the formula's syntax tree.
The $4$ clauses at level $1$ are represented by
$\macdomel_{1},\macdomel_{2},\macdomel_{3}$ and $\macdomel_{4}$.
$\macdomel_{1}$ represents the clause $\left\{\lnot\macpvo \lor \Box
\left[ \macpvt\lor\lnot\macpvh \right] \right\}$. Since $\lnot
\macpvo$ occurs as a literal in this clause, there is a dotted arrow
from $\macdomel_{1}$ to $\macpvo$. $\Box\left[\macpvt\lor \lnot
\macpvh\right]$ (represented by $\macdomel_{9}$) also occurs as a
literal in clause $\macdomel_{1}$ and hence there is an arrow from
$\macdomel_{1}$ to $\macdomel_{9}$. $\macdomel_{4}$ represents the
fourth clause at level $1$, which contains $\Diamond\left[(\macpvf\lor
\lnot\macpvh)\land (\macpvt)\right]$ as a literal. This
$\Diamond\maccnf$ formula is represented by $\macdomel_{10}$. The two
clauses $(\macpvf\lor \lnot\macpvh)$ and $(\macpvt)$ are represented
by $\macdomel_{7}$ and $\macdomel_{8}$ respectively and are connected
to $\macdomel_{10}$ by arrows. The propositional variable $\macpvt$
occurs as literal at $2$ levels, indicated as $\maclev_{0}$ and
$\maclev_{1}$.
\begin{figure}[!htp]
  \begin{center}
    \includegraphics{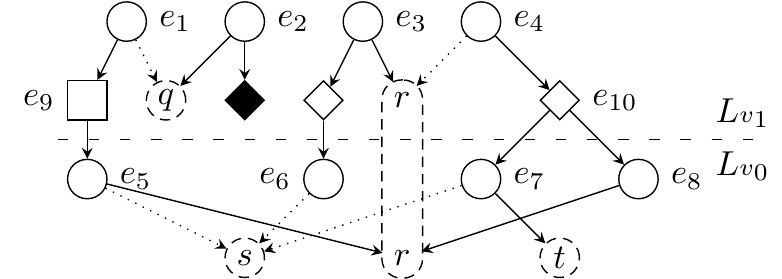}
  \end{center}
  \caption{Relational structure associated with the modal formula
  $\left\{ \lnot\macpvo \lor \Box \left[
  \macpvt\lor\lnot\macpvh \right] \right\}\land \left\{ \macpvo\lor
  \Diamond\bot\right\}\land \left\{ \macpvt\lor \Diamond\left[
  \lnot\macpvh\right] \right\}\land \left\{ \lnot\macpvt
  \lor \Diamond\left[ (\macpvf\lor \lnot\macpvh) \land
  (\macpvt) \right]\right\}$}
  \label{fig:exampleMFRelStruc}
\end{figure}

Now we will formalize the above example. The intuition behind the
following definition is to represent all clauses and literals of a
modal CNF formula by the domain elements of a relational structure.
Binary relations are used to indicate which literals occur in which
clause (and which clauses occur in which literal). Unary relations are
used to indicate which elements represent literals and which elements
represent clauses. This will enable us to reason about clauses,
literals and their dependencies using MSO formulae over the relational
structure.
\begin{definition}
  \label{def:relStrucMF}
  Given a modal CNF formula $\macmfo$, we associate with it a
  relational structure $\macstruc(\macmfo)$. It will have one domain
  element for every $\maccl$ in $\macmfo$. It will have one domain
  element for every literal of the form $\Box\maccl$ or
  $\Diamond\maccnf$ in $\macmfo$. It will also have one domain element
  for every propositional variable used in $\macmfo$. There are no
  domain elements representing the propositional constant $\bot$. They
  will be handled as special cases.
  
  The relational structure will have two binary relations $\macedge$
  (occurs) and $\macedgec$ (occurs negatively).
  $\macedgec(\macdomel_{1},\macdomel_{2})$ iff $\macdomel_{1}$
  represents a $\maccl$ and $\macdomel_{2}$ represents a 
  propositional variable occurring negated as a $\maclit$ in the $\maccl$
  represented by $\macdomel_{1}$. If $\macdomel_{1}$ represents a
  $\maccl$, then $\macedge(\macdomel_{1},\macdomel_{2})$ iff
  $\macdomel_{2}$ represents a literal (occurring in the clause
  represented by $\macdomel_{1}$) of the form $\Box\maccl$,
  $\Diamond\maccnf$ or a non-negated propositional variable. If
  $\macdomel_{1}$ represents a literal of the form $\Box\maccl$, then
  $\macedge(\macdomel_{1},\macdomel_{2})$ iff $\macdomel_{2}$
  represents the corresponding $\maccl$. If $\macdomel_{1}$ represents
  a literal of the form $\Diamond\maccnf$, then
  $\macedge(\macdomel_{1},\macdomel_{2})$ iff $\macdomel_{2}$
  represents a clause in the corresponding $\maccnf$.  Finally, the
  following unary relations are present:\\
  \begin{tabular}[htp]{rcl}
    $\maccls$ & : & contains all domain elements representing
    clauses\\
    $\macclsc$ & : & all domain elements representing literals\\
    $\macuc$ & : & all literals of the form $\Box\bot$\\
    $\macbox$ & : & all literals of the form $\Box\maccl$\\
    $\macdmd$ & : & all literals of the form $\Diamond\maccnf$\\
    $(\maclev_{\macmdidx})_{0 \le \macmdidx \le \macmd(\macmfo)}$ & :
    & all clauses and literals at level $\macmdidx$
  \end{tabular}
\end{definition}
For clauses and literals of the form $\Box\maccl$ or
$\Diamond\maccnf$, there is one domain element for every occurrence of
the clause or literal. For example, if the literal
$\Diamond(\macpvo_{1}\land \macpvo_{2})$ occurs in two different
positions of a big formula $\macmfo$, the two occurrences will be
represented by two different domain elements in $\macstruc(\macmfo)$. In
contrast, different occurrences of a literal that is just a
propositional variable will be represented by the same domain element.
In the rest of the paper, whenever we refer to the treewidth of a modal
CNF formula $\macmfo$, we mean the treewidth of $\macstruc(\macmfo)$.

If $\macdomel_{1}$ represents a $\maccl$,
$\macedge(\macdomel_{1},\macdomel_{2})$ means that the clause
represented by $\macdomel_{1}$ can be satisfied by satisfying
the literal represented by $\macdomel_{2}$.
$\macedgec(\macdomel_{1},\macdomel_{2})$ means that the clause
represented by $\macdomel_{1}$ can be satisfied by setting the
propositional variable represented by $\macdomel_{2}$ to false.

If $\macclsu_{0}\subseteq\maccls\cap\maclev_{0}$ is a subset of domain
elements representing clauses at level 0, let $\maccnf(\macclsu_{0})$
be the modal CNF formula that is the conjunction of clauses
represented by domain elements in $\macclsu_{0}$. We will now see how
to check satisfiability of $\maccnf(\{\macdomel_{7},\macdomel_{8}\})$
in \figref{fig:exampleMFRelStruc} and describe the generalization of
this process given in \eqref{def:modSatInfDescLev0} below. We use
$\macclel$ and $\maclitel$ for first order variables intended to
represent clauses and literals respectively. First of all, there must
be a subset
$\mactlitsu_{0}\subseteq\{\macpvt,\macpvh,\macpvf\}=\macclsc
\cap\maclev_{0}$ that will be set to $\top$, as written in the
beginning of \eqref{def:modSatInfDescLev0}. Then, we must check that
this assignment satisfies each clause $\macclel$ in $\macclsu_{0}$,
written as $\forall \macclel\in\macclsu_{0}$ in
\eqref{def:modSatInfDescLev0}.  To check that the clause represented
by $\macdomel_{7}$ is satisfied, either a positively occurring literal
like $\macpvf$ must be set to $\top$ and hence in $\mactlitsu_{0}$
(written as ``$\exists
\maclitel\in\mactlitsu_{0}:\macedge(\macclel,\maclitel)$'' in
\eqref{def:modSatInfDescLev0}) or a negatively occurring literal like
$\macpvh$ must be set to $\bot$ and hence not in $\mactlitsu_{0}$
(``$\exists \maclitel\in (\macclsc\cap\maclev_{0})\setminus
\mactlitsu_{0}:\macedgec(\macclel,\maclitel)$'' in
\eqref{def:modSatInfDescLev0}). A similar argument applies to
$\macdomel_{8}$ as well.
%  \begin{align*}
%      \maclevsat{0}(\macclsu_{0}) &\macdef \exists \mactlitsu_{0}
%      \subseteq (\macclsc\cap\maclev_{0}):
%      \forall\macclel\in \macclsu_{0}:\\
%       & &\exists\maclitel\in \mactlitsu_{0}&:
%      \macedge(\macclel,\maclitel)\\
%       & &\lor \exists\maclitel\in (\macclsc\cap\maclev_{0})
%      \setminus\mactlitsu_{0}&: \macedgec(\macclel,\maclitel)
%  \end{align*}
\begin{Formula}[def:modSatInfDescLev0]
  \maclevsat{0}(\macclsu_{0}) \macdef \exists \mactlitsu_{0}
  \subseteq
  (\macclsc\cap\maclev_{0}):\forall\macclel\in\macclsu_{0}:
  \>1 \left[\left(\exists
  \maclitel\in\mactlitsu_{0}:\macedge(\macclel,\maclitel)\right)
  \lor \left(\exists \maclitel\in (\macclsc\cap\maclev_{0})\setminus
  \mactlitsu_{0}:\macedgec(\macclel,\maclitel)\right)\right]

  \Form[def:modSatInfDescLevi]
  \maclevsat{\macmdidx}(\macclsu_{\macmdidx}) \macdef \exists
  \mactlitsu_{\macmdidx} \subseteq (\macclsc\cap\maclev_{\macmdidx}):
  \forall \macclel\in \macclsu_{\macmdidx}:
  \>1 \left[\left(\exists \maclitel\in \mactlitsu_{\macmdidx} :
  \macedge(\macclel,\maclitel)\right)
   \lor \left(\exists \maclitel\in (\macclsc\cap\maclev_{\macmdidx})
  \setminus \mactlitsu_{\macmdidx}:
  \macedgec(\macclel,\maclitel)\right)\right]

  \>1 \land [\maccomits_{\macmdidx-1}=\{\macclel'\in
  (\maccls\cap\maclev_{\macmdidx-1})\mid \exists \maclitel'\in
  \mactlitsu_{\macmdidx}\cap\macbox,\macedge(\maclitel',\macclel')\}
  \Rightarrow
  \>2 \forall \maclitel\in \mactlitsu_{\macmdidx}\cap\macdmd:
  \macdemands_{\macmdidx-1}=\{\macclel\in
  (\maccls\cap\maclev_{\macmdidx-1})\mid
  \macedge(\maclitel,\macclel)\} \Rightarrow
  \>3
  \maclevsat{\macmdidx-1}(\macdemands_{\macmdidx-1}\cup
  \maccomits_{\macmdidx-1})]
\end{Formula}

Checking satisfiability at higher levels is slightly more complicated.
Suppose $\macclsu_{\macmdidx}\subseteq \maccls\cap\maclev_{\macmdidx}$
is a subset of clauses at level $\macmdidx$. We will take
$\macclsu_{1}=\{\macdomel_{1},\macdomel_{3},\macdomel_{4}\}$ from
\figref{fig:exampleMFRelStruc} as an example. If some world
$\macworldo$ in some Kripke model $\macmodelo$ satisfies
$\maccnf(\macclsu_{1})$, there must be some subset $\mactlitsu_{1}$ of
literals at level $1$ satisfied at $\macworldo$ (``$\exists
\mactlitsu_{\macmdidx} \subseteq (\macclsc\cap\maclev_{\macmdidx})$''
in \eqref{def:modSatInfDescLevi}). As before, we check that for
every clause represented in $\macclsu_{1}$ (``$\forall \macclel\in
\macclsu_{\macmdidx}$'' in \eqref{def:modSatInfDescLevi}), there is
either a positively occurring literal in $\mactlitsu_{1}$
(``$\exists \maclitel\in \mactlitsu_{\macmdidx} :
\macedge(\macclel,\maclitel)$'' in \eqref{def:modSatInfDescLevi}) or a
negatively occurring literal not in $\mactlitsu_{1}$ (``$\exists
\maclitel\in (\macclsc\cap\maclev_{\macmdidx}) \setminus
\mactlitsu_{\macmdidx}: \macedgec(\macclel,\maclitel)$'' in
\eqref{def:modSatInfDescLevi}). Next, we must check that the literals
we have chosen to be satisfied at $\macworldo$ (by putting them into
$\mactlitsu_{1}$) can actually be satisfied. Suppose
$\mactlitsu_{1}$ was
$\{\macdomel_{9},\macpvo,\macpvt,\macdomel_{10}\}$. Since
$\macdomel_{9}$ represents a literal of the form $\Box\maccl$
(with the $\maccl$ represented by domain element $\macdomel_{5}$), we
are committed to satisfy the clause represented by
$\macdomel_{5}$ in any world succeeding $\macworldo$. Let
$\maccomits_{0}=\{\macdomel_{5}\}$ be the set of clauses occurring at
level $0$ that we have committed to as a result of choosing
corresponding $\Box\maccl$ literals to be in $\mactlitsu_{1}$
(``$\maccomits_{\macmdidx-1}=\{\macclel'\in
(\maccls\cap\maclev_{\macmdidx-1})\mid \exists \maclitel'\in
\mactlitsu_{\macmdidx}\cap\macbox,\macedge(\maclitel',\macclel')\}$''
in \eqref{def:modSatInfDescLevi}). Now, since we have also chosen
$\macdomel_{10}$ to be in $\mactlitsu_{1}$ and $\macdomel_{10}$
represents a $\Diamond\maccnf$ formula, there is a demand to create a
world $\macworldo'$ that succeeds $\macworldo$ and satisfies the
corresponding $\maccnf$ formula. We have to check that every such
demand in $\mactlitsu_{1}$ can be satisfied (``$\forall \maclitel\in
\mactlitsu_{\macmdidx}\cap\macdmd$'' in \eqref{def:modSatInfDescLevi})
by creating successor worlds. In case of the demand created by
$\macdomel_{10}$, $\{\macdomel_{7},\macdomel_{8}\}=\macdemands_{0}$ is
the set of clauses in the demanded $\maccnf$ formula
(``$\macdemands_{\macmdidx-1}=\{\macclel\in
(\maccls\cap\maclev_{\macmdidx-1})\mid
\macedge(\maclitel,\macclel)\}$'' in \eqref{def:modSatInfDescLevi}).
Our aim now is to create a successor world $\macworldo'$ in which all
clauses represented in $\macdemands_{0}$ are satisfied. However,
$\macworldo'$ is a successor world and we have already committed to
satisfying all clauses represented in $\maccomits_{0}$ in all
successor worlds. Hence, we actually check if the clauses represented
in $\maccomits_{0}\cup\macdemands_{0}$ are satisfiable by inductively
invoking $\maclevsat{0}(\macdemands_{0}\cup\maccomits_{0})$
(``$\maclevsat{\macmdidx-1}(\macdemands_{\macmdidx-1}\cup
\maccomits_{\macmdidx-1})$'' in \eqref{def:modSatInfDescLevi}).

For the sake of clarity, we have skipped handling literals of the form
$\Box\bot$ in the above discussion. They will be handled in the formal
arguments that follow. 
\begin{lemma}%[*]
  \label{lem:levSatMSO}
  The property $\maclevsat{\macmdidx}(\macclsu_{\macmdidx})$ can be
  written in a MSO logic formula of size linear in $\macmdidx$. If
  $\macmfo$ is any modal formula in CNF and $\macclsu_{\macmdidx}$ is
  any subset of domain elements representing clauses at level
  $\macmdidx$, then $\maccnf(\macclsu_{\macmdidx})$ is satisfiable iff
  $\maclevsat{\macmdidx}(\macclsu_{\macmdidx})$ is true in
  $\macstruc(\macmfo)$.
\end{lemma}
\begin{details}
  \begin{proof}
    We will prove the first claim by induction on $\macmdidx$. Formula
    \eqref{def:modSatMSODescLev0} below is same as
    \eqref{def:modSatInfDescLev0} written in formal MSO syntax.
    \eqref{def:modSatMSODescLevi} is a formal MSO statement of
    \eqref{def:modSatInfDescLevi} and two additional conditions for
    handling literals of the form $\Box\bot$ and $\Diamond\bot$. We
    will prove that the length $|\maclevsat{\macmdidx}|$ of
    $\maclevsat{\macmdidx}$ is linear in $\macmdidx$. Let $\macico$ be
    the length of $\maclevsat{\macmdidx}$ without length of
    $\maclevsat{\macmdidx-1}$ counted. As can be seen,
    $|\maclevsat{0}|\le \macico$. Inductively assume that
    $|\maclevsat{\macmdidx-1}|\le \macmdidx\macico$.  Then, $|
    \maclevsat{\macmdidx} |= \macico+|\maclevsat{\macmdidx-1}|$.
    Hence, $|\maclevsat{\macmdidx}|\le
    \macico+\macmdidx\macico=\macico(\macmdidx+1)$.

    We will now prove the second claim by induction on $\macmdidx$.

    \emph{Base case $\macmdidx=0$:} The modal formula
    $\maccnf(\macclsu_{0})$ is a propositional CNF formula. Suppose $
    \maclevsat{0}(\macclsu_{0})$ is true in $\macstruc(\macmfo)$.
    Hence, there is a subset $\mactlitsu_{0}$ of domain elements that
    satisfy the last four conditions of $\maclevsat{0}$ defined in
    \eqref{def:modSatMSODescLev0}. The second condition $\forall
    \macfovo\left(\mactlitsu_{0}(x)\Rightarrow \left(
    \macclsc(\macfovo)\land\maclev_{0}(\macfovo)\right)\right)$
    ensures that all domain elements in $\mactlitsu_{0}$ are also in
    $\macclsc$ and $\maclev_{0}$. Hence, all domain elements in
    $\mactlitsu_{0}$ represent literals at level $0$. Since the only
    literals at level $0$ are propositional variables or their
    negations, $\mactlitsu_{0}$ is in fact a subset of propositional
    variables. Consider the Kripke model $\macmodelo$ with a single
    world $\macworldo$ at which, all propositional variables in
    $\mactlitsu_{0}$ are set to $\top$ and all others are set to
    $\bot$. We will now prove that all clauses represented in
    $\macclsu_{0}$ are satisfied in $\macworldo$. Let $\macclel$ be
    some element in $\macclsu_{0}$ representing some clause. Since
    $\macclsu_{0}(\macclel)$ is true and $\macstruc(\macmfo)$
    satisfies the last three conditions of $\maclevsat{0}$, we have
    that either $\exists \maclitel\left(
    \mactlitsu_{0}(\maclitel)\land \macedge(\macclel,\maclitel)
    \right)$ or $\exists \maclitel \left( \macclsc(\maclitel) \land
    \maclev_{0}(\maclitel) \land \lnot\mactlitsu_{0}(\maclitel) \land
    \macedgec(\macclel,\maclitel) \right)$ is true in
    $\macstruc(\macmfo)$. In the first case,
    $\macedge(\macclel,\maclitel)$ means that $\maclitel$ is a
    positively occurring literal in the clause $\macclel$ and
    $\mactlitsu_{0}(\maclitel)$ means that $\maclitel$ is in
    $\mactlitsu_{0}$ (and hence it is set to $\top$ in $\macworldo$,
    satisfying clause $\macclel$). In the second case,
    $\macedgec(\macclel,\maclitel)\land \macclsc(\maclitel)\land
    \maclev_{0}(\maclitel)$ means that $\maclitel$ is a literal
    negatively occurring in clause $\macclel$ and
    $\lnot\mactlitsu_{0}(\maclitel)$ means that $\maclitel$ is not in
    $\mactlitsu_{0}$ (and hence it is set to $\bot$ in $\macworldo$,
    again satisfying clause $\macclel$).

    Now suppose that there is a Kripke model $\macmodelo$ and a world
    $\macworldo$ such that $\macmodelo,\macworldo\models
    \maccnf(\macclsu_{0})$. We will prove that
    $\maclevsat{0}(\macclsu_{0})$ is true in $\macstruc(\macmfo)$. The
    first requirement is to find a suitable subset
    $\mactlitsu_{0}$ of domain elements. We will set
    $\mactlitsu_{0}$ to be the set of precisely those domain elements
    that represent propositional variables occurring at level $0$ and
    set to $\top$ in the world $\macworldo$. This will ensure that the
    condition $\forall \macfovo\left(\mactlitsu_{0}(x)\Rightarrow \left(
    \macclsc(\macfovo)\land\maclev_{0}(\macfovo)\right)\right)$ in
    $\maclevsat{0}$ is satisfied. Now we have to prove that last three
    conditions of $\maclevsat{0}$ are satisfied. So let $\macclel\in
    \macclsu_{0}$ be any domain element so that it satisfies
    $\macclsu_{0}(\macclel)$. We have to now prove that this $\macclel$
    satisfies one of the last two conditions of $\maclevsat{0}$. Since
    $\macclel\in \macclsu_{0}$, it represents a clause in $\macmfo$
    occurring at level $0$. Since $\macmodelo,\macworldo\models
    \maccnf(\macclsu_{0})$, the clause represented by $\macclel$ is
    satisfied in $\macworldo$. Hence there is either a positively
    occurring propositional variable set to $\top$ in $\macworldo$
    (so that it is in $\mactlitsu_{0}$, thus satisfying $\exists
    \maclitel\left( \mactlitsu_{0}(\maclitel)\land
    \macedge(\macclel,\maclitel) \right)$) or a negatively occurring
    propositional variable set to $\bot$ in $\macworldo$ (so that it is
    not in $\mactlitsu_{0}$, thus satisfying $\exists \maclitel \left(
    \macclsc(\maclitel) \land \maclev_{0}(\maclitel) \land
    \lnot\mactlitsu_{0}(\maclitel) \land \macedgec(\macclel,\maclitel)
    \right)$). This completes the base case.
    \begin{Formula}[def:modSatMSODescLev0]
      \maclevsat{0}(\macclsu_{0})\macdef \exists \mactlitsu_{0} 
      \>0\{
      \>1	\forall \macfovo\left(\mactlitsu_{0}(x)\Rightarrow \left(
	  \macclsc(\macfovo)\land\maclev_{0}(\macfovo)\right)\right)
      \>1 \land \forall \macclel \quad \macclsu_{0}(\macclel) \Rightarrow
      \>1[
	\>2 \exists \maclitel\left( \mactlitsu_{0}(\maclitel)\land
	    \macedge(\macclel,\maclitel) \right)
	\>2 \lor \exists \maclitel \left( \macclsc(\maclitel) \land
	    \maclev_{0}(\maclitel) \land \lnot\mactlitsu_{0}(\maclitel)
	    \land \macedgec(\macclel,\maclitel) \right)
      \>1]
      \>0\}
    \end{Formula}

    \emph{Induction step:} Suppose $\macclsu_{\macmdidx}$ is a subset of
    domain elements representing clauses occurring at level $\macmdidx$
    and $\maclevsat{\macmdidx}(\macclsu_{\macmdidx})$ is true in
    $\macstruc(\macmfo)$. We will build a Kripke model $\macmodelo$ and
    prove that it has a world $\macworldo$ such that
    $\macmodelo,\macworldo\models \maccnf(\macclsu_{\macmdidx})$. We
    will start with a single world $\macworldo$. Since
    $\maclevsat{\macmdidx}(\macclsu_{\macmdidx})$ is true in
    $\macstruc(\macmfo)$, there must be a subset
    $\mactlitsu_{\macmdidx}$ of domain elements satisfying the last
    eleven conditions of $\maclevsat{\macmdidx}(\macclsu_{\macmdidx})$.
    The condition $\forall
    \macfovo\left(\mactlitsu_{\macmdidx}(x)\Rightarrow \left(
    \macclsc(\macfovo)\land\maclev_{\macmdidx}(\macfovo)\right)\right)$
    ensures that all domain elements in this $\mactlitsu_{\macmdidx}$
    represent literals occurring at level $\macmdidx$. Let
    $\macpvsul(\mactlitsu_{\macmdidx})\subseteq \mactlitsu_{\macmdidx}$
    be those domain elements in $\mactlitsu_{\macmdidx}$ that
    represent propositional variables. Similarly, let
    $\macbfsul(\mactlitsu_{\macmdidx})$ and
    $\macdfsul(\mactlitsu_{\macmdidx})$ be the domain elements in
    $\mactlitsu_{\macmdidx}$ representing literals of the form
    $\Box\maccl$ and $\Diamond\maccnf$ respectively. In our world
    $\macworldo$, set all propositional variables in
    $\macpvsul(\mactlitsu_{\macmdidx})$ to $\top$ and set all others to
    $\bot$. Now, $\macworldo$ satisfies all literals represented in
    $\macpvsul(\mactlitsu_{\macmdidx})$. We will later prove how to
    satisfy literals represented in $\macbfsul(\mactlitsu_{\macmdidx})$
    and $\macdfsul(\mactlitsu_{\macmdidx})$ in the world $\macworldo$.

    Now, assuming that all literals represented in
    $\mactlitsu_{\macmdidx}$ are satisfied at $\macworldo$, we will
    prove that $\macmodelo,\macworldo\models
    \maccnf(\macclsu_{\macmdidx})$. This part of the proof is similar
    to the base case. If $\macclel$ is any clause in
    $\macclsu_{\macmdidx}$, it satisfies
    $\macclsu_{\macmdidx}(\macclel)$ and hence either $\exists
    \maclitel\left( \mactlitsu_{\macmdidx}(\maclitel)\land
    \macedge(\macclel,\maclitel) \right)$ or $\exists \maclitel \left(
    \macclsc(\maclitel) \land \maclev_{\macmdidx}(\maclitel) \land
    \lnot\mactlitsu_{\macmdidx}(\maclitel) \land
    \macedgec(\macclel,\maclitel) \right)$ is true in
    $\macstruc(\macmfo)$. In the first case, a positively occurring
    literal at level $\macmdidx$ is in $\mactlitsu_{\macmdidx}$, and
    since all literals in $\mactlitsu_{\macmdidx}$ are satisfied at
    $\macworldo$, the clause represented by $\macclel$ is also
    satisfied at $\macworldo$. In the second case, a negatively
    occurring literal at level $\macmdidx$ is not in
    $\mactlitsu_{\macmdidx}$. Since only propositional variables can
    occur negatively in clauses, we can in fact conclude that a
    negatively occurring propositional variable is not in
    $\mactlitsu_{\macmdidx}$. Since all propositional variables not in
    $\mactlitsu_{\macmdidx}$ are set to $\bot$ in $\macworldo$, the
    clause represented by $\macclel$ is satisfied in $\macworldo$.

    \begin{Formula}[def:modSatMSODescLevi]
      \maclevsat{\macmdidx}(\macclsu_{\macmdidx})\macdef \exists \mactlitsu_{\macmdidx} 
      \>0\{
      \>1	\forall \macfovo\left(\mactlitsu_{\macmdidx}(x)\Rightarrow \left(
	  \macclsc(\macfovo)\land\maclev_{\macmdidx}(\macfovo)\right)\right)
      \>1 \land \forall \macclel \quad\macclsu_{\macmdidx}(\macclel) \Rightarrow
      \>1[
	\>2 \exists \maclitel\left( \mactlitsu_{\macmdidx}(\maclitel)\land
	    \macedge(\macclel,\maclitel) \right)
	\>2 \lor \exists \maclitel \left( \macclsc(\maclitel) \land
	    \maclev_{\macmdidx}(\maclitel) \land \lnot\mactlitsu_{\macmdidx}(\maclitel)
	    \land \macedgec(\macclel,\maclitel) \right)
      \>1]
      \>1 \land\exists \macfovo(\mactlitsu_{\macmdidx}(\macfovo)\land\macdmd(\macfovo))
	  \Rightarrow \forall \macfovt\left(\mactlitsu_{\macmdidx}
	  (\macfovt)\Rightarrow \lnot\macuc(\macfovt)\right)
      \>1 \land\forall\macfovo\left( (\mactlitsu_{\macmdidx}(\macfovo)\land
      \macdmd(\macfovo))\Rightarrow \lnot\macuc(\macfovo)\right)
      \>1 \land \exists\maccomits_{\macmdidx-1} 
      \>1 [
	  \>2 \forall\macclel\left(
	      \maccomits_{\macmdidx-1}(\macclel)\Leftrightarrow \exists\maclitel\left( 
	      \mactlitsu_{\macmdidx}(\maclitel)\land \macbox(\maclitel) \land
	      \macedge(\maclitel,\macclel)\right)\right)
	  \>2 \land \forall\maclitel\left( (\mactlitsu_{\macmdidx}(\maclitel)\land
	      \macdmd(\maclitel)) \right)\Rightarrow
	  \>2[
	      \>3 \exists\macdemands_{\macmdidx-1} \forall\macclel'
		  \left(\macdemands_{\macmdidx-1}(\macclel')\Leftrightarrow\left(
		  (\maccomits_{\macmdidx-1}(\macclel')) \lor
		  \macedge(\maclitel,\macclel')  \right)\right)
		  \>4 \land\maclevsat{\macmdidx-1}(\macdemands_{\macmdidx-1})
	  \>2]
      \>1]
      \>0\}
    \end{Formula}

    Now we will prove that literals represented in
    $\macbfsul(\mactlitsu_{\macmdidx})$ and
    $\macdfsul(\mactlitsu_{\macmdidx})$ can be satisfied in $\macworldo$
    by adding appropriate successor worlds. First note that since
    $\forall\macfovo\left( (\mactlitsu_{\macmdidx}(\macfovo)\land
    \macdmd(\macfovo))\Rightarrow \lnot\macuc(\macfovo)\right)$ is true
    in $\macstruc(\macmfo)$, no element $\macfovo$ in
    $\mactlitsu_{\macmdidx}$ represents a literal of the form
    $\Diamond\bot$ (since $\macuc$ is the unary relation containing all
    domain elements representing literals of the form $\Box\bot$ or
    $\Diamond\bot$). Second, note that since $\exists
    \macfovo(\mactlitsu_{\macmdidx}(\macfovo)\land\macdmd(\macfovo))
    \Rightarrow \forall \macfovt\left(\mactlitsu_{\macmdidx}
    (\macfovt)\Rightarrow \lnot\macuc(\macfovt)\right)$ is true in
    $\macstruc(\macmfo)$, if $\macdfsul(\mactlitsu_{\macmdidx})$ is not
    empty, then no element in $\macbfsul(\mactlitsu_{\macmdidx})$
    represents a literal of the form $\Box\bot$. Therefore, we can hope
    to add a new successor for each literal represented by some element
    $\maclitel$ in $\macdfsul(\mactlitsu_{\macmdidx})$, satisfying the
    $\maccnf$ formula in the literal represented in $\maclitel$ as well
    as all clauses in literals represented in
    $\macbfsul(\mactlitsu_{\macmdidx})$. Now we will prove that this can
    actually be done.

    Since $\maclevsat{\macmdidx}(\macclsu_{\macmdidx})$ is true in
    $\macstruc(\macmfo)$, there is a subset $\maccomits_{\macmdidx-1}$
    of domain elements satisfying the last four conditions of
    $\maclevsat{\macmdidx}(\macclsu_{\macmdidx})$. The condition
    $\forall\macclel\left(
    \maccomits_{\macmdidx-1}(\macclel)\Leftrightarrow
    \exists\maclitel\left( \mactlitsu_{\macmdidx}(\maclitel)\land
    \macbox(\maclitel) \land \macedge(\maclitel,\macclel)\right)\right)$
    ensures that $\maccomits_{\macmdidx-1}$ contains exactly those
    domain elements representing some clause $\maccl_{1}$ at level
    $\macmdidx-1$ such that $\Box\maccl_{1}$ is a literal in
    $\mactlitsu_{\macmdidx}$ (and hence $\Box\maccl_{1}$ is in
    $\macbfsul(\mactlitsu_{\macmdidx})$). An element $\maclitel$
    satisfies $\mactlitsu_{\macmdidx}(\maclitel)\land \macdmd(\maclitel)$
    iff $\maclitel\in \macdfsul(\mactlitsu_{\macmdidx})$. Hence, the
    condition $\forall\maclitel\left(
    (\mactlitsu_{\macmdidx}(\maclitel)\land \macdmd(\maclitel))
    \right)\Rightarrow[\cdots]$ ensures that the last two conditions of
    $\maclevsat{\macmdidx}(\macclsu_{\macmdidx})$ is true for every
    element $\maclitel$ in $\mactlitsu_{\macmdidx}$ representing a literal of the form
    $\Diamond\maccnf$ occurring at level $\macmdidx$. Consider any one
    such element $\maclitel$. The condition
    $\exists\macdemands_{\macmdidx-1} \forall\macclel'
    \left(\macdemands_{\macmdidx-1}(\macclel')\Leftrightarrow\left(
    (\maccomits_{\macmdidx-1}(\macclel')) \lor
    \macedge(\maclitel,\macclel')  \right)\right)$ ensures that
    $\macdemands_{\macmdidx-1}$ contains exactly those elements
    representing some clause $\maccl_{1}$ (occurring at level
    $\macmdidx-1$) such that $\Box\maccl_{1}$ is represented in
    $\macbfsul(\mactlitsu_{\macmdidx})$ or $\maccl_{1}$ occurs in the
    $\maccnf$ formula in the $\Diamond\maccnf$ literal represented by
    $\maclitel$. Since
    $\maclevsat{\macmdidx-1}(\macdemands_{\macmdidx-1})$ is true, we can
    apply the induction hypothesis and conclude that there is some Kripke
    model $\macmodelo'$ and a world $\macworldo'$ such that
    $\macmodelo',\macworldo'\models\maccnf(\macdemands_{\macmdidx-1})$.
    Now, $\macworldo'$ satisfies the $\maccnf$ formula in the
    $\Diamond\maccnf$ literal represented by $\maclitel$. For every
    literal of the form $\Box\maccl_{1}$ in
    $\macbfsul(\mactlitsu_{\macmdidx})$, $\macworldo'$ satisfies
    $\maccl_{1}$. Now, we add the Kripke model $\macmodelo'$ to
    $\macmodelo$ and make $\macworldo'$ a successor of $\macworldo$. We
    repeat this procedure for every element $\maclitel$ in
    $\macdfsul(\mactlitsu_{\macmdidx})$. Now, for every literal in
    $\mactlitsu_{\macmdidx}$ of the form $\Diamond\maccnf$, there is a
    successor of $\macworldo$ that satisfies the corresponding $\maccnf$
    formula (we have already proved that literals of the form
    $\Diamond\bot$ will not be present in $\mactlitsu_{\macmdidx}$). For
    every literal in $\mactlitsu_{\macmdidx}$ of the form $\Box\maccl$,
    all successors of $\macworldo$ will satisfy the corresponding
    $\maccl$ (we have already proved that if there is a literal of the
    form $\Box\bot$ in $\mactlitsu_{\macmdidx}$, then
    $\mactlitsu_{\macmdidx}$ will not have any literals of the form
    $\Diamond\maccnf$ and hence we will not add any successor worlds to
    $\macworldo$).

    Now we will prove the other direction of the induction step.
    Suppose $\macclsu_{\macmdidx}$ is a subset of domain elements
    representing clauses occurring at level $\macmdidx$ and that there
    is a Kripke model $\macmodelo$ and a world $\macworldo$ such that
    $\macmodelo,\macworldo\models \maccnf(\macclsu_{\macmdidx})$. We
    will prove that $\maclevsat{\macmdidx}(\macclsu_{\macmdidx})$ is
    true in $\macstruc(\macmfo)$. To begin with, we will choose
    $\mactlitsu_{\macmdidx}$ to be the set of precisely those domain
    elements that represent literals occurring at level $\macmdidx$ that
    are satisfied at $\macworldo$. If literals of the form $\Box\bot$
    occur at level $\macmdidx$, then they will also be included in
    $\mactlitsu_{\macmdidx}$ by definition if there are no successor
    worlds at $\macworldo$. Now we will prove that last eleven
    conditions of $\maclevsat{\macmdidx}(\macclsu_{\macmdidx})$ are true
    in $\macstruc(\macmfo)$. The condition $\forall
    \macfovo\left(\mactlitsu_{\macmdidx}(x)\Rightarrow \left(
    \macclsc(\macfovo)\land\maclev_{\macmdidx}(\macfovo)\right)\right)$
    is true since all elements $\macfovo$ in $\mactlitsu_{\macmdidx}$ are
    representing literals ($\macclsc(\macfovo)$) at level $\macmdidx$
    ($\maclev_{\macmdidx}(\macfovo)$). Next we will prove that the
    condition $\forall \macclel \macclsu_{\macmdidx}(\macclel)
    \Rightarrow [\dots]$ is true. Let $\macclel$ be some arbitrary
    element in $\macclsu_{\macmdidx}$. Since $\macclel$ represents a
    clause that is satisfied at the world $\macworldo$, there must be
    either a positively occurring literal that is satisfied at
    $\macworldo$ (and hence the domain element representing that literal
    will be in $\mactlitsu_{\macmdidx}$, thus implying that $\exists
    \maclitel\left( \mactlitsu_{\macmdidx}(\maclitel)\land
    \macedge(\macclel,\maclitel) \right)$ is true in
    $\macstruc(\macmfo)$) or there must be a negatively occurring
    literal that is not satisfied at $\macworldo$. In the latter case,
    since only propositional variables can occur negatively, we can in
    fact conclude that there is a negatively occurring propositional
    variable that is set to $\bot$ at $\macworldo$ (and hence not in
    $\mactlitsu_{\macmdidx}$), which implies that $\exists \maclitel
    \left( \macclsc(\maclitel) \land \maclev_{\macmdidx}(\maclitel)
    \land \lnot\mactlitsu_{\macmdidx}(\maclitel) \land
    \macedgec(\macclel,\maclitel) \right)$ is true in
    $\macstruc(\macmfo)$.

    Next we will prove that the condition $\forall\macfovo\left(
    (\mactlitsu_{\macmdidx}(\macfovo)\land \macdmd(\macfovo))\Rightarrow
    \lnot\macuc(\macfovo)\right)$ is true. If any element $\macfovo$ is
    in $\mactlitsu_{\macmdidx}$ ($\mactlitsu_{\macmdidx}(\macfovo)$) and
    represents a literal of the form $\Diamond\maccnf$
    ($\macdmd(\macfovo)$), then $\macfovo$ will not represent
    $\Diamond\bot$ ($\lnot\macuc(\macfovo)$) since $\macfovo$ represents
    a literal that is satisfied at $\macworldo$ and $\Diamond\bot$
    cannot be satisfied.

    Next we will prove that the condition $\exists
    \macfovo(\mactlitsu_{\macmdidx}(\macfovo)\land\macdmd(\macfovo))
    \Rightarrow \forall \macfovt\left(\mactlitsu_{\macmdidx}
    (\macfovt)\Rightarrow \lnot\macuc(\macfovt)\right)$ is true. Suppose
    there is some element $\macfovo$ in $\mactlitsu_{\macmdidx}$ that
    represents a literal of the form $\Diamond\maccnf$ ($\exists
    \macfovo(\mactlitsu_{\macmdidx}(\macfovo)\land\macdmd(\macfovo))$).
    Since the literal represented by $\macfovo$ is satisfied at
    $\macworldo$, there is a successor world in which the corresponding
    $\maccnf$ formula is satisfied. Since $\macworldo$ has successor worlds, it
    cannot satisfy $\Box\bot$ and hence none of the elements in
    $\mactlitsu_{\macmdidx}$ represent literals of the form $\Box\bot$
    ($\forall \macfovt\left(\mactlitsu_{\macmdidx}
    (\macfovt)\Rightarrow \lnot\macuc(\macfovt)\right)$).
    
    Finally, we will prove that the condition
    $\exists\maccomits_{\macmdidx-1}[\dots]$ is true. Let us first
    construct the set $\maccomits_{\macmdidx-1}$. For any element
    $\macclel$ ($\forall \macclel$), we will put $\macclel$ in
    $\maccomits_{\macmdidx-1}$ iff
    ($\maccomits_{\macmdidx-1}(\macclel)\Leftrightarrow$) there is some
    element $\maclitel$ ($\exists \maclitel$) in $\mactlitsu_{\macmdidx}$
    ($\mactlitsu_{\macmdidx}(\maclitel)$) representing a literal of the
    form $\Box\maccl$ ($\macbox(\maclitel)$) such that $\macclel$
    represents the corresponding $\maccl$
    ($\macedge(\maclitel,\macclel)$). The condition
    $\forall\macclel\left(
    \maccomits_{\macmdidx-1}(\macclel)\Leftrightarrow
    \exists\maclitel\left( \mactlitsu_{\macmdidx}(\maclitel)\land
    \macbox(\maclitel) \land \macedge(\maclitel,\macclel)\right)\right)$
    is true in $\macstruc(\macmfo)$ by construction. Next we will prove
    that the condition $\forall\maclitel\left(
    (\mactlitsu_{\macmdidx}(\maclitel)\land \macdmd(\maclitel))
    \right)\Rightarrow[\dots]$ is true. Suppose $\maclitel$ is any
    element in $\mactlitsu_{\macmdidx}$ representing a literal of the
    form $\Diamond\maccnf_{1}$ ($\macdmd(\maclitel)$). Since
    $\Diamond\maccnf_{1}$ is satisfied at the world $\macworldo$, there is a
    successor world $\macworldo'$ that satisfies the corresponding
    $\maccnf_{1}$ formula. Let $\macdemands_{\macmdidx-1}$ be the set
    that includes some element $\macclel'$ iff ($\forall\macclel'
    \macdemands_{\macmdidx-1}(\macclel')\Leftrightarrow$) $\macclel'$ is
    in the set $\maccomits_{\macmdidx-1}$ constructed above
    ($\maccomits_{\macmdidx-1}(\macclel')$) or it represents a $\maccl$
    occurring in the $\maccnf_{1}$ formula contained in the
    $\Diamond\maccnf_{1}$ literal represented by $\maclitel$
    ($\macedge(\maclitel,\macclel')$). $\macdemands_{\macmdidx-1}$
    satisfies the condition $\forall\macclel'
    \left(\macdemands_{\macmdidx-1}(\macclel')\Leftrightarrow\left(
    (\maccomits_{\macmdidx-1}(\macclel')) \lor
    \macedge(\maclitel,\macclel')  \right)\right)$ by construction. If
    $\macclel'$ is any element in $\macdemands_{\macmdidx-1}$, then it
    represents some $\maccl_{1}$ at level $\macmdidx-1$ such that
    $\maccl_{1}$ occurs in the $\maccnf_{1}$ formula contained in the
    $\Diamond\maccnf_{1}$ literal represented by $\maclitel$ or
    $\Box\maccl_{1}$ appears in $\mactlitsu_{\macmdidx}$. Hence, all
    clauses represented in $\macdemands_{\macmdidx-1}$ are satisfied at
    $\macworldo'$ (since $\macworldo'$ is a successor of $\macworldo$
    that satisfies $\maccnf_{1}$ and all literals of the form
    $\Box\maccl$ represented in $\mactlitsu_{\macmdidx}$ are satisfied
    at $\macworldo$). By the induction hypothesis, we conclude that
    $\maclevsat{\macmdidx-1}(\macdemands_{\macmdidx-1})$ is true in
    $\macstruc(\macmfo)$.\qed
  \end{proof}
\end{details}
\begin{theorem}
  \label{thm:modalSatFpt}
  Given a modal CNF formula $\macmfo$, there is a \macfpt\/ algorithm
  that checks if $\macmfo$ is satisfiable in general models, with
  treewidth of $\macstruc(\macmfo)$ and modal depth of $\macmfo$ as
  parameters.
\end{theorem}
\begin{proof}
  Given $\macmfo$, $\macstruc(\macmfo)$ can be constructed in
  polynomial time. To check that all clauses of $\macmfo$ at level
  $\macmd(\macmfo)$ are satisfiable in some world $\macworldo$ of some
  Kripke model $\macmodelo$, we check whether the formula $\exists
  \macclsu_{\macmd(\macmfo)} \forall
  \macclel(\macclsu_{\macmd(\macmfo)}(\macclel)\Leftrightarrow
  (\maccls(\macclel)\land \maclev_{\macmd(\macmfo)}(\macclel)))\land
  \maclevsat{\macmd(\macmfo)}(\macclsu_{\macmd(\macmfo)})$ is true in
  $\macstruc(\macmfo)$. By \lemref{lem:levSatMSO}, this is possible
  iff $\macmfo$ is satisfiable and length of the above formula is
  linear in $\macmd(\macmfo)$. An application of Courcelle's theorem
  will give us the \macfpt\/ algorithm.\qed
\end{proof}

\begin{details}
\subsection{On the relevance of treewidth for modal logic}
Informally, treewidth is a measure of how close a graph is to being a
tree. Given a modal logic formula $\macmfo$, the associated structure
$\macstruc(\macmfo)$ is very similar to the syntax tree of $\macmfo$.
The structure $\macstruc(\macmfo)$ is not a tree (i.e., it has cycles)
because a single propositional variable may be shared by many clauses
of the formula. Thus, if very few variables are shared across clauses,
$\macstruc(\macmfo)$ is very close to a tree, i.e.,
$\macstruc(\macmfo)$ will have small treewidth. In the example of
\figref{fig:exampleMFRelStruc}, if we replace $\macpvo$ and $\macpvh$
by $\macpvt$, the number of shared variables will increase. As can be
seen in \figref{fig:exampleMFRelStruc1}, the number of cycles will
also increase.
\begin{figure}[!htp]
  \begin{center}
    \includegraphics{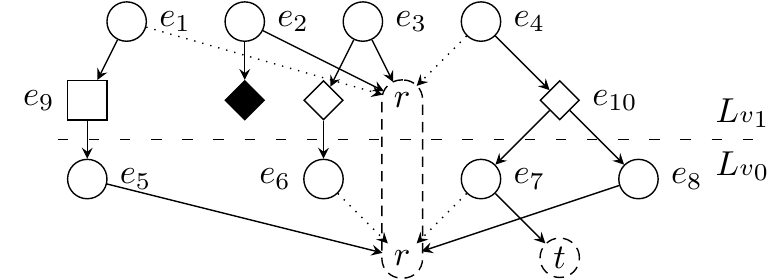}
  \end{center}
  \caption{Relational structure associated with the modal formula
  $\left\{ \lnot\macpvt \lor \Box \left[ \macpvt
  \right] \right\}\land \left\{ \macpvt\lor \Diamond\bot\right\}\land
  \left\{ \macpvt\lor \Diamond\left[ \lnot\macpvt\right] \right\}\land
  \left\{ \lnot\macpvt \lor \Diamond\left[ (\macpvf\lor \lnot\macpvt)
  \land (\macpvt) \right]\right\}$}
  \label{fig:exampleMFRelStruc1}
\end{figure}
For example, $\macdomel_{1}$ was not part of any cycle in
\figref{fig:exampleMFRelStruc} but forms a cycle with
$\macdomel_{9}, \macdomel_{5}$ and $\macpvt$ in
\figref{fig:exampleMFRelStruc1}.

Treewidth is a very fundamental concept and naturally arises in many
contexts, even in industrial applications like software verification
\cite{Thorup98}. Applications of treewidth related techniques to
propositional logic is extensively studied --- see \cite[Section
1.4]{FMR08} and references therein. Modal logic being a natural and
very useful extension of propositional logic, we might expect some
benefit by exploring applicability of treewidth related techniques to
modal logic.

The set of modal formulas with small treewidth is powerful enough to
encode complex formulas. In \cite[Lemma 1]{ALM09}, there is a
translation of propositional CNF formulae into equivalent modal
formulae. We can verify that the resulting modal formula always has a
small constant treewidth (the resulting modal formula uses only one
propositional variable).  Hence, the restriction of bounded treewidth
is not a severe one. Given a formula $\macmfo$, $\macstruc(\macmfo)$
can be computed in \macp\/.  Though computing treewidth of
$\macstruc(\macmfo)$ is \macnp\/-complete, it is \macfpt\/ when
parameterized by treewidth.
\end{details}

\section{Models with Euclidean property}
\label{sec:modSatEuclFrames}
In this section, we will investigate the parameterized complexity of
satisfiability in Euclidean models. The main observation leading to
the \macfpt\/ algorithm is the fact that if a modal formula is
satisfied in a Euclidean model, then it is satisfied in a rather
simple model. As proved in \cite{KP09}, if a modal formula is
satisfied at some world $\macworldo_{0}$ in some Euclidean model
$\macmodelo$, then it is satisfied in a model whose underlying frame
is of the form $(\macworlds\cup \{ \macworldo_{0}\},\macrelo)$ where
$\macworlds\times \macworlds \subseteq\macrelo$. Therefore, almost all
worlds are successors of almost all other worlds. If one world
satisfies a formula $\Box\maccl_{1}$, then almost all worlds satisfy
the formula $\maccl_{1}$ (and hence satisfy $\Box\maccl_{1}$ as well).
If one world satisfies a formula $\Diamond\maccnf_{1}$, then almost
all worlds satisfy $\Diamond\maccnf_{1}$ as well. Thus, most of the
worlds are very similar to each other and we can reason about them
using small MSO formulae.  This holds even if we add more properties
like reflexivity, transitivity etc.
\begin{conf}
  The technical details needed for the following result can be found in
  the full version.
\end{conf}
\begin{details}
  The rest of this section is devoted to proving the following
  theorem.
\end{details}
\begin{theorem}
  \label{thm:modSatEuclFPT}
  Let $\macmfo$ be a modal CNF formula. With treewidth of
  $\macstruc(\macmfo)$ as parameter, there is a \macfpt\/
  algorithm for checking whether $\macmfo$ is satisfiable in a
  Kripke model that satisfies Euclidean property and any combination
  of reflexivity, symmetry and transitivity.
\end{theorem}
\begin{minorext}
  We will drop all unary relations $(\maclev_{\macmdidx})_{0\le
  \macmdidx\le \macmd(\macmfo)}$. Instead, we will have one unary
  relation $\macpvs$ containing all domain elements representing
  propositional variables and one unary relation $\machle$ containing
  all other domain elements.  This will not change the treewidth of
  $\macstruc(\macmfo)$. 
  To make the presentation easier to follow, we will use informal
  description of MSO formulae. Let $\macclsu_{1}$ and $\macglclsu$ be sets
  of clauses (we will see later that clauses
  in $\macglclsu$ will be satisfied in almost all worlds of a model).
  The following MSO formula checks if all clauses in $\macclsu_{1}$
  are satisfiable in a model in which, all worlds satisfy all clauses
  in $\macglclsu$.
  \begin{Formula}[def:mdGlSatInfDesc]
    \macmdglsat(\macclsu_{1},\macglclsu) \macdef \exists \mactlitsu
    \subseteq (\macclsc \cap \macdmd):
    \exists \mactlitsu_{0}\subseteq \macpvs:
    \>1 \macgllitsu=\left\{ \maclitel\in (\macclsc\cap \macbox) \mid \exists
    \macclel\in \macglclsu \land \macedge(\maclitel,\macclel)\right\}
    \>1 \land \forall \macclel\in (\macclsu_{1}\cup \macglclsu):
      \>2 \exists \maclitel\in (\macgllitsu\cup \mactlitsu \cup
      \mactlitsu_{0}): \macedge(\macclel,\maclitel)
      \>2 \lor \exists \maclitel\in \macpvs\setminus
      \mactlitsu_{0}: \macedgec(\macclel,\maclitel)
    \>1 \land \forall \maclitel\in \mactlitsu: \exists
    \mactlitsu_{1} \subseteq \macpvs:
      \>2 \macdemands=\{ \macclel'\in \maccls\mid
      \macedge(\maclitel,\macclel')\} \Rightarrow
      \>2 \forall \macclel\in (\macdemands\cup \macglclsu):
        \>3 \exists \maclitel'\in (\mactlitsu\cup \macgllitsu\cup
	\mactlitsu_{1}): \macedge(\macclel,\maclitel')
	\>3 \lor \exists \maclitel'\in \macpvs\setminus
	\mactlitsu_{1}: \macedgec(\macclel,\maclitel')
  \end{Formula}
  \begin{lemma}
    \label{lem:mdGlSatMSO}
    Let $\macclsu_{1}$ and $\macglclsu$ be sets of clauses occurring in a
    modal CNF formula $\macmfo$. If $\macmdglsat(\macclsu_{1},\macglclsu)$
    is true in $\macstruc(\macmfo)$, then there is a Kripke model
    $\macmodelo$ and a world $\macworldo$ in it such that:
    \begin{enumerate}
      \item $\macworldo$ satisfies all clauses in $\macclsu_{1}$,
      \item all worlds in $\macmodelo$ satisfy all clauses in
	$\macglclsu$,
      \item the accessibility relation $\macrelo$ in $\macmodelo$ is the
	equivalence relation on the set of all worlds in $\macmodelo$
	and
      \item if a new world satisfying all clauses in $\macglclsu$ is
	added to $\macmodelo$ (making it accessible from all existing
	worlds of $\macmodelo$), $\macworldo$ will still satisfy all
	clauses in $\macclsu_{1}$.
    \end{enumerate}   
  \end{lemma}
  \begin{proof}
    Suppose $\macmdglsat(\macclsu_{1},\macglclsu)$ is true in
    $\macstruc(\macmfo)$. We will build a Kripke model $\macmodelo$
    satisfying the required properties. To begin with, there must be a
    set $\mactlitsu$ of literals of the form $\Diamond\maccnf$ and
    a set $\mactlitsu_{0}$ of propositional variables such that
    the rest of the formula $\macmdglsat(\macclsu_{1},\macglclsu)$ is true in
    $\macstruc(\macmfo)$. Let $\macgllitsu=\left\{ \maclitel\in
    (\macclsc\cap \macbox) \mid \exists \macclel\in \macglclsu \land
    \macedge(\maclitel,\macclel)\right\}$ be the set of literals of the
    form $\Box\maccl$ such that the corresponding clause is in
    $\macglclsu$. We will start with one world $\macworldo$ in
    our model in which precisely those propositional variables are set
    to $\top$ that are in the set $\mactlitsu_{0}$. We will then add
    exactly one world $\macworldo_{\maclitidx}$ for each literal
    $\maclitel_{\maclitidx}$ in $\mactlitsu$. For any literal
    $\maclitel_{\maclitidx}$ in $\mactlitsu$, there will be a subset
    $\mactlitsu_{1}$ of propositional variables such that last four
    conditions of $\macmdglsat(\macclsu_{1},\macglclsu)$ are true in
    $\macstruc(\macmfo)$. In the world $\macworldo_{\maclitidx}$, we will
    set precisely those propositional variables to $\top$ that are in
    the set $\mactlitsu_{1}$ corresponding to $\maclitel_{\maclitidx}$.
    Our model consists of the above worlds and the accessibility
    relation $\macrelo$ is the equivalence relation on the set of all
    worlds. For any literal $\maclitel_{\maclitidx}\in \mactlitsu$
    (which is of the form $\Diamond\maccnf$), let
    $\macdemands_{\maclitidx}=\{ \macclel'\in \maccls\mid
    \macedge(\maclitel_{\maclitidx},\macclel')\}$ be the set of clauses
    that make up the corresponding $\maccnf$ formula. By induction on
    modal depth of any clause $\macclel$, we will prove that if
    $\macclel\in \macclsu_{1} \cup \macglclsu$, then $\macclel$ is satisfied
    in $\macworldo$ and if $\macclel\in \macdemands_{\maclitidx}\cup
    \macglclsu$, then $\macclel$ is satisfied in
    $\macworldo_{\maclitidx}$.

    In the base case, modal depth of $\macclel$ is $0$. We will first
    prove that if $\macclel\in \macclsu_{1} \cup \macglclsu$, then
    $\macclel$ is satisfied in $\macworldo$. Suppose $\exists
    \maclitel\in (\macgllitsu\cup \mactlitsu \cup \mactlitsu_{0}):
    \macedge(\macclel,\maclitel)$ is true in $\macstruc(\macmfo)$. If
    $\maclitel\in \macgllitsu\cup \mactlitsu$, then $\maclitel$ will
    have modal depth at least $1$ (because it is of the form
    $\Box\maccl$ or $\Diamond\maccnf$) and hence modal depth of
    $\macclel$ (which contains $\maclitel$ as a sub-formula) will also
    be more than $1$. Hence, $\maclitel\in \mactlitsu_{0}$. This means
    that $\maclitel$ is a propositional variable set to $\top$ in
    $\macworldo$ that occurs positively in $\macclel$ and hence
    $\macclel$ is satisfied in $\macworldo$. If $\exists \maclitel\in
    \macpvs\setminus \mactlitsu_{0}: \macedgec(\macclel,\maclitel)$ is
    true in $\macstruc(\macmfo)$, then there is a propositional variable
    set to $\bot$ in $\macworldo$ that occurs negatively in $\macclel$
    and hence, $\macclel$ is satisfied in $\macworldo$. Now, we will
    take up the case of $\macclel\in \macdemands_{\maclitidx}\cup
    \macglclsu$. In the formula $\macmdglsat(\macclsu_{1},\macglclsu)$,
    suppose $\exists \maclitel'\in (\mactlitsu\cup \macgllitsu\cup
    \mactlitsu_{1}): \macedge(\macclel,\maclitel')$ is true in
    $\macstruc(\macmfo)$. As before, $\maclitel'$ has be to in
    $\mactlitsu_{1}$, which means that there is a propositional variable
    set to $\top$ at $\macworldo_{\maclitidx}$ that occurs positively in
    $\macclel$. Hence, $\macclel$ is satisfied in
    $\macworldo_{\maclitidx}$. If on the other hand, $\exists
    \maclitel'\in \macpvs\setminus \mactlitsu_{1}:
    \macedgec(\macclel,\maclitel')$ is true, then there is a
    propositional variable set to $\bot$ at $\macworldo_{\maclitidx}$
    that occurs negatively in $\macclel$. Hence, $\macclel$ is satisfied
    in $\macworldo_{\maclitidx}$.

    For the induction step, suppose $\macclel\in (\macclsu_{1}\cup
    \macglclsu)$. Suppose that in $\macstruc(\macmfo)$, the formula
    $\exists \maclitel\in (\macgllitsu\cup \mactlitsu \cup
    \mactlitsu_{0}): \macedge(\macclel,\maclitel)$ is true. If
    $\maclitel\in \macgllitsu$, then it is of the form $\Box\maccl$ such
    that the corresponding clause (of modal depth lower than $\macclel$)
    is in $\macglclsu$. By the induction hypothesis, all worlds in
    $\macmodelo$ will satisfy all clauses in $\macglclsu$ of modal depth
    less than $\macclel$ (and by condition 4 of the lemma, all new
    worlds added will also satisfy all clauses in $\macglclsu$) and
    hence $\Box\maccl$ is satisfied in $\macworldo$ and hence
    $\macclel$ is satisfied in $\macworldo$. If $\maclitel\in
    \mactlitsu$, then $\maclitel$ is a literal of the form
    $\Diamond\maccnf$ such that there is a world
    $\macworldo_{\maclitidx}$ in $\macmodelo$ added to satisfy the
    corresponding $\maccnf$ formula. All clauses in this $\maccnf$
    formula (which form the set $\macdemands_{\maclitidx}$) have modal
    depth less than $\macclel$ and by the induction hypothesis, they are
    all satisfied at $\macworldo_{\maclitidx}$.  Hence,
    $\macworldo_{\maclitidx}$ satisfies the corresponding $\maccnf$
    formula, and hence $\macworldo$ satisfies the corresponding
    $\Diamond\maccnf$ formula (since $\macworldo_{\maclitidx}$ is a
    successor of $\macworldo$) and hence $\macclel$ is satisfied at
    $\macworldo$. If $\maclitel\in \mactlitsu_{0}$, then $\maclitel$
    is a propositional variable set to $\top$ in $\macworldo$ and that
    occurs positively in $\macclel$.  Hence, $\macclel$ is satisfied
    in $\macworldo$. On the other hand, if $\exists \maclitel\in
    \macpvs\setminus \mactlitsu_{0}: \macedgec(\macclel,\maclitel)$ is
    true in $\macstruc(\macmfo)$, then $\maclitel$ is a propositional
    variable set to $\bot$ in $\macworldo$ and that occurs negatively
    in $\macclel$. Hence, in this case also, $\macclel$ is satisfied
    in $\macworldo$.

    Finally, for the induction step, suppose $\macclel\in
    \macdemands_{\maclitidx}\cup \macglclsu$. Suppose that in
    $\macstruc(\macmfo)$, the formula $\exists \maclitel'\in
    (\mactlitsu \cup \macgllitsu\cup \mactlitsu_{1}):
    \macedge(\macclel,\maclitel')$ is true. If $\maclitel'\in
    \macgllitsu$, then it is of the form $\Box\maccl$ such that the
    corresponding clause (of modal depth lower than $\macclel$) is in
    $\macglclsu$. By the induction hypothesis, all worlds in $\macmodelo$
    will satisfy all clauses in $\macglclsu$ of modal depth less than
    $\macclel$ (and by condition 4 of the lemma, all new worlds added
    will also satisfy all clauses in $\macglclsu$) and hence
    $\Box\maccl$ is satisfied in $\macworldo_{\maclitidx}$ and hence
    $\macclel$ is satisfied in $\macworldo_{\maclitidx}$. If
    $\maclitel'\in \mactlitsu$, then $\maclitel'$ is a literal of the
    form $\Diamond\maccnf$ such that there is a world
    $\macworldo_{\maclitidx'}$ in $\macmodelo$ added to satisfy the
    corresponding $\maccnf$ formula. All clauses in this $\maccnf$
    formula (which form the set $\macdemands_{\maclitidx'}$) have
    modal depth less than $\macclel$ and by the induction hypothesis, they
    are all satisfied at $\macworldo_{\maclitidx'}$.  Hence,
    $\macworldo_{\maclitidx'}$ satisfies the corresponding $\maccnf$
    formula, and hence $\macworldo_{\maclitidx}$ satisfies the
    corresponding $\Diamond\maccnf$ formula (since
    $\macworldo_{\maclitidx'}$ is a successor of
    $\macworldo_{\maclitidx}$) and hence $\macclel$ is satisfied at
    $\macworldo_{\maclitidx}$. If $\maclitel'\in \mactlitsu_{1}$, then
    $\maclitel'$ is a propositional variable set to $\top$ in
    $\macworldo_{\maclitidx}$ and that occurs positively in
    $\macclel$.  Hence, $\macclel$ is satisfied in
    $\macworldo_{\maclitidx}$. On the other hand, if $\exists
    \maclitel'\in \macpvs\setminus \mactlitsu_{1}:
    \macedgec(\macclel,\maclitel')$ is true in $\macstruc(\macmfo)$,
    then $\maclitel'$ is a propositional variable set to $\bot$ in
    $\macworldo_{\maclitidx}$ and that occurs negatively in
    $\macclel$. Hence, in this case also, $\macclel$ is satisfied in
    $\macworldo_{\maclitidx}$. This completes the induction step and
    hence the proof.\qed
  \end{proof}

  The following formula makes use of
  $\macmdglsat(\macclsu_{1},\macglclsu)$ to check if a set of clauses
  $\macclsu_{0}$ is satisfiable in an Euclidean model.
  \begin{Formula}[def:mdEuclSatInfDesc]
    \macmdglsat(\macclsu_{0}) \macdef 
    \exists \mactlitsu_{0}\subseteq\macclsc: \forall \maclitel\in
    \mactlitsu_{0}: \exists \macclel\in \macclsu_{0}:
    \macedge(\macclel,\maclitel)
      \>1 \land \forall \macclel\in \macclsu_{0}:
      \>2 \exists \maclitel\in \mactlitsu_{0}:
      \macedge(\macclel,\maclitel)
      \>2 \lor \exists \maclitel\in \macpvs\setminus
      \mactlitsu_{0}: \macedgec(\macclel,\maclitel)
      \>1 \land \maccomits_{0}=\{\macclel\in \maccls\mid \exists
      \maclitel\in (\mactlitsu_{0}\cap \macbox)\land
      \macedge(\maclitel,\macclel)\} \Rightarrow
      \>2 \exists \macglclsu\subseteq (\maccls): \forall
      \macclel\in \macglclsu: \exists \maclitel\in (\macclsc\cap
      \macbox): \macedge(\maclitel,\macclel)
      \>2 \land \forall \maclitel\in (\mactlitsu_{0}\cap \macdmd):
      \>3 \macdemands_{0}=\{\macclel\in \maccls\mid
      \macedge(\maclitel,\macclel)\}\Rightarrow
      \>3 \macmdglsat(\macdemands_{0}\cup
      \maccomits_{0},\macglclsu)
  \end{Formula}
  \begin{lemma}
    \label{lem:mdEuclSat}
    Let $\macclsu_{0}$ be a set of clauses occurring in a modal CNF
    formula $\macmfo$. $\maccnf(\macclsu_{0})$ is satisfiable at a
    world $\macworldo$ in an Euclidean model $\macmodelo$ in which
    $\macworldo$ is not its own successor iff
    $\macmdglsat(\macclsu_{0})$ is true in $\macstruc(\macmfo)$.
  \end{lemma}
  \begin{proof}
    Suppose $\macmdglsat(\macclsu_{0})$ is true in $\macstruc(\macmfo)$.
    We will build an Euclidean Kripke model $\macmodelo$ satisfying
    $\maccnf(\macclsu_{0})$. We will begin with a single world
    $\macworldo$. In $\macworldo$, precisely those propositional
    variables are set to $\top$ that appear in the set
    $\mactlitsu_{0}$ that witnesses truth of
    $\macmdglsat(\macclsu_{0})$ in $\macstruc(\macmfo)$. For now, we
    will assume that literals of the form $\Diamond\maccnf$ and
    $\Box\maccl$ in $\mactlitsu_{0}$ will be satisfied in $\macworldo$
    by addition of suitable worlds. If $\macclel$ is any clause in
    $\macclsu_{0}$, then either $\exists \maclitel\in \mactlitsu_{0}:
    \macedge(\macclel,\maclitel)$ or $\exists \maclitel\in
    \macpvs\setminus \mactlitsu_{0}: \macedgec(\macclel,\maclitel)$ is
    true. In the former case, a propositional variable that is set to
    $\top$ at $\macworldo$ or a literal of the form $\Diamond\maccnf$
    or $\Box\maccl$ that is satisfied in $\macworldo$ occurs in
    $\macclel$ and hence $\macclel$ is satisfied in $\macworldo$. In the
    later case, a propositional variable set to $\bot$ in $\macworldo$
    occurs negatively in $\macclel$ and hence $\macclel$ is satisfied
    in $\macworldo$.

    As promised, we will now add suitable successors such that all
    literals of the form $\Diamond\maccnf$ and $\Box\maccl$ in
    $\mactlitsu_{0}$ are satisfied at $\macworldo$. Let
    $\maccomits_{0}=\{\macclel\in \maccls\mid \exists \maclitel\in
    (\mactlitsu_{0}\cap \macbox)\land \macedge(\maclitel,\macclel)\}$ be the
    set of clauses we have committed to satisfy in all successors of
    $\macworldo$ by choosing the corresponding $\Box\maccl$ to be in
    $\mactlitsu_{0}$. For any literal $\maclitel$ of the form
    $\Diamond\maccnf$ in $\mactlitsu_{0}\cap \macdmd$, let
    $\macdemands_{0}=\{\macclel\in \maccls\mid
    \macedge(\maclitel,\macclel)\}$ be the set of clauses in the
    corresponding $\maccnf$ formula. Since
    $\macmdglsat(\macdemands_{0}\cup \maccomits_{0},\macglclsu)$ is true
    in $\macstruc(\macmfo)$, there is a model $\macmodelo_{1}$ and a
    world $\macworldo_{1}$ in it as specified in
    \lemref{lem:mdGlSatMSO}, such that $\macworldo_{1}$ satisfies all
    clauses in $\macdemands_{0}\cup \maccomits_{0}$ and all worlds in
    $\macmodelo_{1}$ satisfy all clauses in $\macglclsu$. Let
    $(\macmodelo_{1},\macworldo_{1}),
    (\macmodelo_{2},\macworldo_{2}),\dots$ be the models given by
    \lemref{lem:mdGlSatMSO} for the demand sets created by each of the
    literals of the form $\Diamond\maccnf$ in $\mactlitsu_{0}\cap
    \macdmd$. Adding all worlds of $\macmodelo_{1},\macmodelo_{2},\dots$
    to $\macmodelo$ and making $\macworldo_{1},\macworldo_{2},\dots$
    successors of $\macworldo$ will result in all literals of the form
    $\Diamond\maccnf$ in $\mactlitsu_{0}\cap \macdmd$ being satisfied at
    $\macworldo$ in $\macmodelo$. Since
    $\macworldo_{1},\macworldo_{2},\dots$ all satisfy all clauses in
    $\maccomits_{0}$, all literals of the form $\Box\maccl$ in
    $\mactlitsu_{0}\cap \macbox$ are also satisfied at $\macworldo$ in
    $\macmodelo$. Making all worlds other than $\macworldo$ successors
    of all worlds other than $\macworldo$ will ensure that $\macmodelo$
    is based on an Euclidean frame. Condition 4 of
    \lemref{lem:mdGlSatMSO} will ensure that due to the additional
    accessibility relation pairs created, the worlds
    $\macworldo_{1},\macworldo_{2},\dots$ will not stop satisfying
    clauses required to satisfy the demands created by literals in
    $\mactlitsu_{0}$.

    Now, suppose that $\maccnf(\macclsu_{0})$ is satisfied in an
    Euclidean model. We will prove that $\macmdglsat(\macclsu_{0})$ is
    true in $\macstruc(\macmfo)$. As proved in \cite{KP09},
    $\maccnf(\macclsu_{0})$ is satisfied in a model $\macmodelo$ at a
    world $\macworldo$ such that the underlying frame of $\macmodelo$ is
    of the form $(\macworlds\cup \{\macworldo\},\macrelo)$ such that
    $\macworlds\times \macworlds\subseteq \macrelo$. As stated in the
    lemma, $\macworldo$ is not its own successor. If $\macworldo$ was
    the successor of any other world, then Euclidean property will
    force $\macworldo$ to be its own successor, hence $\macworldo$
    can not be the successor of any other world. To prove that
    $\macmdglsat(\macclsu_{0})$ is true in $\macstruc(\macmfo)$, we will
    first construct a set $\mactlitsu_{0}$ of literals. Since
    $\maccnf(\macclsu_{0})$ is satisfied at $\macworldo$, for every
    clause in $\macclsu_{0}$, there must be a literal occurring in that
    clause satisfied at $\macworldo$. Let $\mactlitsu_{0}$ be the set of
    such literals of the form $\Box\maccl$ or $\Diamond\maccnf$ and the
    propositional variables set to $\top$ in $\macworldo$ and occurring
    positively in some clause in $\macclsu_{0}$. The condition
    $\forall \maclitel\in \mactlitsu_{0}: \exists \macclel\in
    \macclsu_{0}: \macedge(\macclel,\maclitel)$ is true by construction
    of $\mactlitsu_{0}$. If $\macclel$ is any clause in
    $\macclsu_{0}$, then either
    \begin{enumerate}
      \item there is some literal of the form $\Box\maccl$ or
	$\Diamond\maccnf$ or a positively occurring propositional
	variable that occurs in $\macclel$ and present in
	$\mactlitsu_{0}$ (in which case $\exists \maclitel\in
	\mactlitsu_{0}: \macedge(\macclel,\maclitel)$ is true) or
      \item there is a negatively occurring propositional variable that
	is set to $\bot$ in $\macworldo$ (in which case $\exists
	\maclitel\in \macpvs\setminus \mactlitsu_{0}:
	\macedgec(\macclel,\maclitel)$ is true).
    \end{enumerate}

    Let $\maccomits_{0}=\{\macclel\in \maccls\mid \exists \maclitel\in
    (\mactlitsu_{0}\cap \macbox)\land \macedge(\maclitel,\macclel)\}$
    be the set of clauses that are satisfied in all successors of
    $\macworldo$.  Let $\macglclsu=\{\macclel\in \maccls\mid \exists
    \maclitel\in (\macclsc\cap \macbox)\land
    \macedge(\maclitel,\macclel) \land \macmodelo,\macworldo'\models
    \maclitel, \macworldo'\ne \macworldo\}$ be the set of all clauses
    such that corresponding $\Box\maccl$ formula is satisfied at some
    world $\macworldo'$ other than $\macworldo$. Since all worlds
    other than $\macworldo$ are successors of $\macworldo'$, all
    worlds other than $\macworldo$ satisfy all clauses in
    $\macglclsu$. The condition $\forall \macclel\in \macglclsu:
    \exists \maclitel\in (\macclsc\cap \macbox):
    \macedge(\maclitel,\macclel)$ is true by construction of
    $\macglclsu$. Note that $\macglclsu$ and $\mactlitsu_{0}\cap
    \macdmd$ will be empty if $\macworldo$ has no successors, hence
    the rest of $\macmdglsat(\macclsu_{0})$ is vacuously true. For any
    literal $\maclitel\in \mactlitsu_{0}\cap \macdmd$, let
    $\macdemands_{0}=\{\macclel\in \maccls\mid
    \macedge(\maclitel,\macclel)\}$ be the set of clauses in the
    $\maccnf$ formula contained in $\maclitel$. In $\macmodelo$, there
    is a successor $\macworldo_{1}$ of $\macworldo$ that satisfies all
    clauses in $\macdemands_{0}\cup \maccomits_{0}$. We will prove
    that $\macmdglsat(\macdemands_{0}\cup \maccomits_{0},\macglclsu)$
    is true in $\macstruc(\macmfo)$.

    We first select a subset $\mactlitsu\subseteq (\macclsc\cap
    \macdmd)$ so that the rest of the formula
    $\macmdglsat(\macdemands_{0}\cup \maccomits_{0},\macglclsu)$ can be
    satisfied. Let $\mactlitsu=\{\maclitel\in (\macclsc\cap \macdmd)\mid
    \macmodelo,\macworldo'\models\maclitel,\macworldo'\ne \macworldo\}$
    be the set of literals of the form $\Diamond\maccnf$ such that some
    world $\macworldo'$ other than $\macworldo$ satisfies the
    $\Diamond\maccnf$ formula (since $\macworldo$ is not a successor
    of $\macworldo'$, some other world $\macworldo''$ succeeding
    $\macworldo'$ will satisfy the corresponding $\maccnf$ formula).
    Let $\mactlitsu_{0}$ be the set of propositional variables set to
    $\top$ in the world $\macworldo_{1}$ mentioned above. Let
    $\macgllitsu=\left\{ \maclitel\in (\macclsc\cap \macbox) \mid
    \exists \macclel\in \macglclsu \land
    \macedge(\maclitel,\macclel)\right\}$ be the set of literals of
    the form $\Box\maccl$ such that the corresponding clause is in
    $\macglclsu$. Let $\macclel$ be any clause in $\macdemands_{0}\cup
    \maccomits_{0}\cup \macglclsu$. Since $\macworldo_{1}$ satisfies
    $\macclel$, there must be a literal $\maclitel$ occurring in
    $\macclel$ such that $\maclitel$ is satisfied in $\macworldo_{1}$.
    \begin{enumerate}
      \item If $\maclitel$ is of the form $\Box\maccl$, then it is in
	$\macgllitsu$ and hence $\exists \maclitel\in \macgllitsu:
	\macedge(\macclel,\maclitel)$ is true.
      \item If $\maclitel$ is of the form $\Diamond\maccnf$, then it is
	in $\mactlitsu$ and hence $\exists \maclitel\in \mactlitsu:
	\macedge(\macclel,\maclitel)$ is true.
      \item If $\maclitel$ is a positively occurring propositional
	variable, then $\exists \maclitel\in \mactlitsu_{0}:
	\macedge(\macclel,\maclitel)$ is true.
      \item If $\maclitel$ is a negatively occurring propositional
	variable, then $\exists \maclitel\in \macpvs\setminus
	\mactlitsu_{0}: \macedgec(\macclel,\maclitel)$ is true.
    \end{enumerate}
    Let $\maclitel\in \mactlitsu$ be any literal of the form
    $\Diamond\maccnf$ in $\mactlitsu$. By definition of $\mactlitsu$,
    there is some world $\macworldo'$ other than $\macworldo$ such that
    $\macworldo'$ satisfies the corresponding $\maccnf$ formula. Let
    $\mactlitsu_{1}$ be the set of propositional variables set to
    $\top$ in $\macworldo'$ and let $\macdemands=\{\macclel'\in
    \maccls\mid \macedge(\maclitel,\macclel')\}$ be the set of clauses
    in the $\maccnf$ formula in $\maclitel$. Let $\macclel$ be any
    clause in $\macdemands\cup \macglclsu$. Since $\macworldo'$
    satisfies $\macclel$,  there must be a literal
    $\maclitel'$ occurring in $\macclel$ such that $\maclitel'$ is
    satisfied in $\macworldo'$.
    \begin{enumerate}
      \item If $\maclitel'$ is of the form $\Box\maccl$, then it is in
	$\macgllitsu$ and hence $\exists \maclitel'\in \macgllitsu:
	\macedge(\macclel,\maclitel')$ is true.
      \item If $\maclitel'$ is of the form $\Diamond\maccnf$, then it is
	in $\mactlitsu$ and hence $\exists \maclitel'\in \mactlitsu:
	\macedge(\macclel,\maclitel')$ is true.
      \item If $\maclitel'$ is a positively occurring propositional
	variable, then $\exists \maclitel'\in \mactlitsu_{1}:
	\macedge(\macclel,\maclitel')$ is true.
      \item If $\maclitel'$ is a negatively occurring propositional
	variable, then $\exists \maclitel'\in \macpvs\setminus
	\mactlitsu_{1}: \macedgec(\macclel,\maclitel')$ is true.
    \end{enumerate}\qed
  \end{proof}

  Suppose a modal formula is satisfied at a world $\macworldo$ in an
  Euclidean model where $\macworldo$ is its own successor. Then
  Euclidean property will force the accessibility relation $\macrelo$ to
  be the equivalence relation on the set of all worlds.  Hence, any
  $\Box\maccl$ literal chosen to be satisfied in $\macworldo$ will
  result in all worlds (including $\macworldo$) satisfying the
  corresponding clause. This can be easily handled by modifying
  $\macmdglsat(\macclsu_{0})$ as follows.
  \begin{Formula}[def:mdEuclReflSatInfDesc]
    \macmdglsat'(\macclsu_{0}) \macdef 
    \exists \mactlitsu_{0}\subseteq\macclsc: 
      \>1 \exists \macglclsu\subseteq (\maccls): \forall
      \macclel\in \macglclsu: \exists \maclitel\in (\macclsc\cap
      \macbox): \macedge(\maclitel,\macclel)
      \>1 \land \forall \macclel\in (\macclsu_{0}\cup \macglclsu):
      \>2 \exists \maclitel\in \mactlitsu_{0}:
      \macedge(\macclel,\maclitel)
      \>2 \lor \exists \maclitel\in \macpvs\setminus
      \mactlitsu_{0}: \macedgec(\macclel,\maclitel)
      \>1 \land \forall \maclitel\in (\mactlitsu_{0}\cap \macdmd):
      \>2 \macdemands_{0}=\{\macclel\in \maccls\mid
      \macedge(\maclitel,\macclel)\}\Rightarrow
      \>2 \macmdglsat(\macdemands_{0},\macglclsu)
  \end{Formula}
  Now to check if a modal formula $\macmfo$ is satisfiable in an
  Euclidean model, we just have to check if
  $\macmdglsat(\macclsu_{0})\lor \macmdglsat'(\macclsu_{0})$ is true
  in $\macstruc(\macmfo)$, where $\macclsu_{0}$ is the set of clauses
  at the highest level. An application of Courcelle's theorem will
  give us the \macfpt{} algorithm. Note that in this case, the size of
  the MSO formula we need to check is independent of modal depth.

  To check if a modal formula $\macmfo$ is satisfiable in a reflexive
  and Euclidean model, we just check if
  $\macmdglsat'(\macclsu_{0})$ is true in $\macstruc(\macmfo)$.

  Suppose a modal formula $\macmfo$ is satisfied at some world
  $\macworldo$ in an Euclidean and symmetric model. If $\macworldo$
  has any other successors, then Euclidean property will force all
  worlds reachable from $\macworldo$ to be successors of $\macworldo$
  and $\macworldo$ to be a successor of all worlds reachable from
  $\macworldo$. This is same as a reflexive and Euclidean model and
  can be handled by $\macmdglsat'(\macclsu_{0})$. If $\macworldo$ has
  no other successors but is its own successor it can again be handled
  by $\macmdglsat'(\macclsu_{0})$. If $\macworldo$ has no successors
  and is not its own successor, then all clauses of $\macmfo$ at the
  highest level are satisfied at $\macworldo$ by literals of the form
  $\Box\maccl$ or propositional variables. This can be easily checked
  by a small MSO formula.

  \subsection{Euclidean and transitive models}
  \label{sec:modSatEuclTransFrames}
  Suppose we want to check satisfiability of a modal CNF formula in
  models that are both Euclidean and transitive. As seen above, the
  modal CNF formula is satisfied in a model with an underlying frame of
  the form $(\macworlds\cup \{\macworldo\},\macrelo)$ where
  $\macworlds\times \macworlds\subseteq \macrelo$. In addition, all
  other worlds are successors of $\macworldo$. Hence any literal of
  the form $\Box\maccl$ satisfied at $\macworldo$ will result in all
  other worlds satisfying the corresponding clause. This can be
  handled by modifying $\macmdglsat(\macclsu)$ as follows.
  \begin{Formula}[def:mdEuclTransSatInfDesc]
    \macmdglsat''(\macclsu_{0}) \macdef 
    \exists \mactlitsu_{0}\subseteq\macclsc: \forall \maclitel\in
    \mactlitsu_{0}: \exists \macclel\in \macclsu_{0}:
    \macedge(\macclel,\maclitel)
      \>1 \land \forall \macclel\in \macclsu_{0}:
      \>2 \exists \maclitel\in \mactlitsu_{0}:
      \macedge(\macclel,\maclitel)
      \>2 \lor \exists \maclitel\in \macpvs\setminus
      \mactlitsu_{0}: \macedgec(\macclel,\maclitel)
      \>1 \land \maccomits_{0}=\{\macclel\in \maccls\mid \exists
      \maclitel\in (\mactlitsu_{0}\cap \macbox)\land
      \macedge(\maclitel,\macclel)\} \Rightarrow
      \>2 \exists \macglclsu\subseteq (\maccls): \forall
      \macclel\in \macglclsu: \exists \maclitel\in (\macclsc\cap
      \macbox): \macedge(\maclitel,\macclel)
      \>2 \land \forall \maclitel\in (\mactlitsu_{0}\cap \macdmd):
      \>3 \macdemands_{0}=\{\macclel\in \maccls\mid
      \macedge(\maclitel,\macclel)\}\Rightarrow
      \>3 \macmdglsat(\macdemands_{0}\cup
      \maccomits_{0},\macglclsu\cup \maccomits_{0})
  \end{Formula}

  The Euclidean property is very strong in the sense that it makes the
  complexity of infinitely many modal logics drop from
  \macpspace{}-hard to \macnp{}-complete \cite{HR07}. One might hope
  for extending the results of this section to any modal logic whose
  frames is a subset of Euclidean frames. The results in \cite{HR07}
  use semantic characterizations while our MSO formulae can only
  reason about syntax of modal logic formulae. Even though there is a
  close relation between the syntax and semantics of modal logic of
  Euclidean frames (which have been used to obtain the results of this
  section), it seems difficult to exploit this relation to obtain
  \macfpt{} algorithms for arbitrary extensions of modal logic of
  Euclidean frames.  It remains to be seen if other tools from the
  theory of MSO logic on graphs can be used to achieve this.

  \section{Reflexive models}
  \label{sec:modSatReflFrames}
  As an example of how the basic technique described in
  \secref{sec:modSatGenFrames} can be extended to satisfiability in
  models satisfying some other properties, we will show satisfiability
  in reflexive models. We will need the following MSO formula to define
  the set of vertices reachable from a given vertex in a finite directed
  acyclic graph.
  \begin{Formula}[eq:reachInfMSO]
    \macreach(\macfovo, \macsovo) \macdef \forall \macfovt \quad
    (\macfovt \in \macsovo) \Leftrightarrow
    \>1[
      \>2 \macedge(\macfovo, \macfovt)
      \>2 \lor \exists \macfovh \in \macsovo : \macedge( \macfovh,
      \macfovt)
    \>1]
  \end{Formula}
  \begin{lemma}
    \label{lem:reachInfMSO}
    Let $\macgraph$ be a finite directed acyclic graph in which,
    $\macedge$ is the binary relation represented by the directed edges.
    Let $\macfovo$ be a vertex and $\macsovo$ be a subset of vertices in
    $\macgraph$. Then, $\macsovo$ is the set of precisely those vertices
    reachable from $\macfovo$ by a directed path of length $1$ or more
    iff $\macreach( \macfovo, \macsovo)$ defined in
    \eqref{eq:reachInfMSO} is true in $\macgraph$.
  \end{lemma}
  \begin{proof}
    Suppose $\macsovo$ is the set of precisely those vertices
    reachable from $\macfovo$ by a directed path of length $1$ or more.
    If some vertex $\macfovt$ is in $\macsovo$ (i.e., there is a
    directed path of length $1$ or more from $\macfovo$ to $\macfovt$),
    we will prove that $\macedge(\macfovo, \macfovt) \lor \exists
    \macfovh \in \macsovo : \macedge( \macfovh, \macfovt)$ is true in
    $\macgraph$. If the length of the path from $\macfovo$ to
    $\macfovt$ is $1$, there is an edge from $\macfovo$ to
    $\macfovt$, making $\macedge( \macfovo, \macfovt)$ true in
    $\macgraph$. If the directed path between $\macfovo$ and $\macfovt$
    is at least $2$ and is of the form $\macfovo \rightarrow \macfovt'
    \rightarrow \cdots \rightarrow \macfovt'' \rightarrow \macfovt$,
    then we can take $\macfovt''$ as witness for $\macfovh$ in
    $\exists \macfovh \in \macsovo : \macedge( \macfovh, \macfovt)$,
    making $\exists \macfovh \in \macsovo : \macedge( \macfovh,
    \macfovt)$ true in $\macgraph$. On the other hand, suppose
    $\macedge(\macfovo, \macfovt) \lor \exists \macfovh \in \macsovo
    : \macedge( \macfovh, \macfovt)$ is true in $\macgraph$ for some
    vertex $\macfovt$. We will prove that $\macfovt$ is in $\macsovo$
    (i.e., there is a directed path of length $1$ or more from
    $\macfovo$ to $\macfovt$). Suppose $\macedge(\macfovo, \macfovt)$ is
    true in $\macgraph$. Then there is an edge from $\macfovo$ to
    $\macfovt$, which is a directed path of length $1$. Suppose
    $\exists \macfovh \in \macsovo
    : \macedge( \macfovh, \macfovt)$ is true $\macgraph$, then there is
    a directed path of length $1$ or more from $\macfovo$ to $\macfovh$
    (since $\macfovh$ is in $\macsovo$). Appending the edge from
    $\macfovh$ to $\macfovt$ to this path gives us a path of length
    $2$ or more from $\macfovo$ to $\macfovt$.

    Now, suppose that $\forall \macfovt \quad (\macfovt \in \macsovo)
    \Leftrightarrow [ \macedge(\macfovo, \macfovt) \lor (\exists
    \macfovh \in \macsovo : \macedge( \macfovh, \macfovt)) ]$ is true in
    $\macgraph$. We will prove that $\macsovo$ is the set of precisely
    those vertices reachable from $\macfovo$ by a directed path of
    length $1$ or more. We will first prove that $\macsovo$ does not
    contain any vertex not reachable from $\macfovo$. Suppose to the
    contrary that there is a vertex $\macfovt$ in $\macsovo$ not
    reachable from $\macfovo$. Since $\macfovt \in \macsovo$,
    $\macedge(\macfovo, \macfovt) \lor (\exists \macfovh \in \macsovo :
    \macedge( \macfovh, \macfovt))$ is true in $\macgraph$. Since
    $\macedge( \macfovo, \macfovt)$ cannot be true (as that would mean
    $\macfovt$ is reachable from $\macfovo$), there is some
    $\macfovh_{1} \in \macsovo$ with $\macedge( \macfovh_{1},
    \macfovt)$. Since $\macfovh_{1}$ is in $\macsovo$,
    $\macedge(\macfovo, \macfovh_{1}) \lor (\exists \macfovh \in \macsovo :
    \macedge( \macfovh, \macfovh_{1}))$ is true.  Since
    $\macedge( \macfovo, \macfovh_{1})$ cannot be true (as that would mean
    $\macfovt$ is reachable from $\macfovo$), there is some
    $\macfovh_{2} \in \macsovo$ with $\macedge( \macfovh_{2},
    \macfovh_{1})$. The vertex $\macfovh_{2}$ has to be distinct from
    $\macfovt$ and $\macfovh_{1}$ since otherwise, the fact that
    $\macgraph$ is devoid of directed cycles is violated. Continuing
    this type of reasoning leads us to an infinite sequence $\macfovt,
    \macfovh_{1}, \macfovh_{2}, \dots$ of distinct vertices,
    contradicting the fact that $\macgraph$ is a finite graph. Hence,
    $\macsovo$ does not contain any vertex not reachable from
    $\macfovo$. Next we will prove that every vertex $\macfovt$
    reachable from $\macfovo$ is in $\macsovo$, by induction on the
    length $\macmdidx$ of the shortest directed path from $\macfovo$ to
    $\macfovt$. In the base case $\macmdidx=1$, there is an edge from
    $\macfovo$ to $\macfovt$, which means that $\macedge( \macfovo,
    \macfovt)$ is true in $\macgraph$, and $\macreach( \macfovo,
    \macsovo)$ forces $\macfovt$ to be in $\macsovo$. Suppose there is a
    directed path of length $\macmdidx+1$ from $\macfovo$ to $\macfovt$.
    Let $\macfovh$ be the vertex preceding $\macfovt$ in this path.
    Since there is a directed path of length at most $\macmdidx$ from
    $\macfovo$ to $\macfovh$, we can use induction hypothesis conclude
    that $\macfovh \in \macsovo$. Since there is an edge from $\macfovh$
    to $\macfovt$, $\exists \macfovh \in \macsovo : \macedge( \macfovh,
    \macfovt)$ is true in $\macgraph$, and again $\macreach( \macfovo,
    \macsovo)$ forces $\macfovt$ to be in $\macsovo$.\qed
  \end{proof}

  Let $\macclsu_{\macmdidx}$ be some set of domain elements representing
  clauses at level at most $\macmdidx$. The property
  $\macmdsat{\macmdidx}( \macclsu_{\macmdidx})$ defined below checks if
  there is a reflexive Kripke model $\macmodelo$ and a world
  $\macworldo$ in it that satisfies all clauses in
  $\macclsu_{\macmdidx}$.
  \begin{Formula}[def:modSatInfDescMaxLev0]
    \macmdsat{0}(\macclsu_{0}) \macdef \exists \mactlitsu_{0}
    \subseteq
    (\macclsc\cap\maclev_{0}):\forall\macclel\in\macclsu_{0}:
    \>1 \left[\left(\exists
    \maclitel\in\mactlitsu_{0}:\macedge(\macclel,\maclitel)\right)
    \lor \left(\exists \maclitel\in (\macclsc\cap\maclev_{0})\setminus
    \mactlitsu_{0}:\macedgec(\macclel,\maclitel)\right)\right]

    \Form[def:modSatInfDescMaxLevi]
    \macmdsat{\macmdidx}(\macclsu_{\macmdidx}) \macdef \exists
    \mactlitsu_{\macmdidx} \subseteq \macclsc:
    \>1 \forall \maclitel\in \mactlitsu_{\macmdidx} \quad \exists \macclel \in
    \macclsu_{\macmdidx}: \exists \macsovo: (\macreach( \macclel,
    \macsovo) \land \maclitel \in \macsovo)
    \>1 \land \maccomits_{\macmdidx-1}=\{\macclel'\in
    \maccls\mid \exists \maclitel'\in
    \mactlitsu_{\macmdidx}\cap\macbox,\macedge(\maclitel',\macclel')\}
    \Rightarrow
    \>1 \forall \macclel\in \macclsu_{\macmdidx} \cup
    \maccomits_{\macmdidx-1}:
    \>2 \left[\left(\exists \maclitel\in \mactlitsu_{\macmdidx} :
    \macedge(\macclel,\maclitel)\right)
     \lor \left(\exists \maclitel\in \macclsc
    \setminus \mactlitsu_{\macmdidx}:
    \macedgec(\macclel,\maclitel)\right)\right]

    \>1 \land \forall \maclitel\in \mactlitsu_{\macmdidx}\cap\macdmd:
    \macdemands_{\macmdidx-1}=\{\macclel\in
    \maccls\mid
    \macedge(\maclitel,\macclel)\} \Rightarrow
    \>3
    \macmdsat{\macmdidx-1}(\macdemands_{\macmdidx-1}\cup
    \maccomits_{\macmdidx-1})
  \end{Formula}
  \begin{lemma}
    \label{lem:mdSatMSO}
    The property $\macmdsat{\macmdidx}(\macclsu_{\macmdidx})$ can be
    written in a MSO logic formula of size linear in $\macmdidx$. If
    $\macmfo$ is any modal formula in CNF and $\macclsu_{\macmdidx}$ is
    any subset of domain elements representing clauses at level
    at most $\macmdidx$, then $\maccnf(\macclsu_{\macmdidx})$ is
    satisfiable in a reflexive model iff
    $\macmdsat{\macmdidx}(\macclsu_{\macmdidx})$ is true in
    $\macstruc(\macmfo)$.
  \end{lemma}
  \begin{proof}
    We will prove the first claim by induction on $\macmdidx$. We
    will prove that the length $|\macmdsat{\macmdidx}|$ of
    $\macmdsat{\macmdidx}$ is linear in $\macmdidx$. Let $\macico$ be
    the length of $\macmdsat{\macmdidx}$ without the length of
    $\macmdsat{\macmdidx-1}$ counted. As can be seen,
    $|\macmdsat{0}|\le \macico$. Inductively assume that
    $|\macmdsat{\macmdidx-1}|\le \macmdidx\macico$.  Then, $|
    \macmdsat{\macmdidx} |= \macico+|\macmdsat{\macmdidx-1}|$.
    Hence, $|\macmdsat{\macmdidx}|\le
    \macico+\macmdidx\macico=\macico(\macmdidx+1)$.

    We will now prove the second claim by induction on $\macmdidx$.

    \emph{Base case $\macmdidx=0$:} Suppose $
    \macmdsat{0}(\macclsu_{0})$ is true in $\macstruc(\macmfo)$.
    Hence, there is a subset $\mactlitsu_{0}$ of domain elements that
    satisfy all the conditions of $\macmdsat{0}$ defined in
    \eqref{def:modSatInfDescMaxLev0}.  Since all domain elements in
    $\mactlitsu_{0}$ represent literals at level $0$ and the only
    literals at level $0$ are propositional variables or their
    negations, $\mactlitsu_{0}$ is in fact a subset of propositional
    variables. Consider the reflexive Kripke model $\macmodelo$ with a single
    world $\macworldo$ at which, all propositional variables in
    $\mactlitsu_{0}$ are set to $\top$ and all others are set to
    $\bot$. We will now prove that all clauses represented in
    $\macclsu_{0}$ are satisfied in $\macworldo$. Let $\macclel$ be
    some element in $\macclsu_{0}$ representing some clause. We have
    that either $\exists \maclitel \in \mactlitsu_{0} : \macedge(\macclel,\maclitel)$
    or $\exists \maclitel\in (\macclsc\cap\maclev_{0})\setminus
    \mactlitsu_{0}:\macedgec(\macclel,\maclitel)$ is true in
    $\macstruc(\macmfo)$. In the first case, a positively occurring
    propositional variable is set to $\top$ in $\macworldo$ and in the
    second case, a negatively occurring propositional variable is set to
    $\bot$ in $\macworldo$.

    Now suppose that there is a reflexive Kripke model $\macmodelo$ and
    a world $\macworldo$ such that $\macmodelo,\macworldo\models
    \maccnf(\macclsu_{0})$. We will prove that
    $\macmdsat{0}(\macclsu_{0})$ is true in $\macstruc(\macmfo)$. The
    first requirement is to find a suitable subset $\mactlitsu_{0}$ of
    domain elements. We will set $\mactlitsu_{0}$ to be the set of
    precisely those domain elements that represent propositional
    variables occurring at level $0$ and set to $\top$ in the world
    $\macworldo$. Since every clause $\macclel$ in $\macclsu_{0}$ is
    satisfied in $\macworldo$, either there is a positively occurring
    propositional variable set to $\top$ in $\macworldo$ or there is a
    negatively occurring propositional variable set to $\bot$ in
    $\macworldo$. In the first case $\exists \maclitel \in
    \mactlitsu_{0} : \macedge(\macclel,\maclitel)$ is true and in the
    second case $\exists \maclitel\in (\macclsc\cap\maclev_{0})\setminus
    \mactlitsu_{0}:\macedgec(\macclel,\maclitel)$ is true in
    $\macstruc( \macmfo)$. This completes the base case.

    \emph{Induction step:} Suppose $\macclsu_{\macmdidx}$ is a subset of
    domain elements representing clauses occurring at level at most
    $\macmdidx$ and $\macmdsat{\macmdidx}(\macclsu_{\macmdidx})$ is true
    in $\macstruc(\macmfo)$. We will build a reflexive Kripke model
    $\macmodelo$ and prove that it has a world $\macworldo$ such that
    $\macmodelo,\macworldo\models \maccnf(\macclsu_{\macmdidx})$. We
    will start with a single world $\macworldo$. Since
    $\macmdsat{\macmdidx}(\macclsu_{\macmdidx})$ is true in
    $\macstruc(\macmfo)$, there must be a subset
    $\mactlitsu_{\macmdidx}$ of domain elements satisfying all the
    conditions of $\macmdsat{\macmdidx}(\macclsu_{\macmdidx})$. Since
    $\forall \maclitel\in \mactlitsu_{\macmdidx} \quad \exists \macclel \in
    \macclsu_{\macmdidx}: \exists \macsovo: (\macreach( \macclel,
    \macsovo) \land \maclitel \in \macsovo)$ is true in
    $\macstruc( \macmfo)$, all literals in $\mactlitsu_{\macmdidx}$ are
    reachable from some clause in $\macclsu_{\macmdidx}$ in
    $\macstruc( \macmfo)$. Hence, all literals in
    $\mactlitsu_{\macmdidx}$ are at level at most $\macmdidx$. Let
    $\maccomits_{\macmdidx-1}=\{\macclel'\in \maccls\mid \exists
    \maclitel'\in
    \mactlitsu_{\macmdidx}\cap\macbox,\macedge(\maclitel',\macclel')\}$
    be the set of clauses that we are committed to satisfy in all
    successors of $\macworldo$ (that includes $\macworldo$ as well) as a
    result of chosing the corresponding $\Box \maccl$ to be in
    $\mactlitsu_{\macmdidx}$. All clauses in $\maccomits_{\macmdidx-1}$
    are at level at most $\macmdidx-1$. For each literal
    $\maclitel_{1}$ of the form $\Diamond \maccnf$ in
    $\mactlitsu_{\macmdidx}$, let $\macdemands_{\macmdidx-1}=\{\macclel\in
    \maccls\mid
    \macedge(\maclitel_{1},\macclel)\}$ be the set of clauses occurring in
    $\maclitel_{1}$. Since all clauses in $\macdemands_{\macmdidx-1}
    \cup \maccomits_{\macmdidx-1}$ are at level at most $\macmdidx-1$
    and $\macmdsat{\macmdidx-1}( \macdemands_{\macmdidx-1} \cup
    \maccomits_{\macmdidx-1})$ is true in $\macstruc ( \macmfo)$, we can
    apply induction hypothesis to conclude that there is a reflexive
    Kripke model $\macmodelo_{1}$ and a world $\macworldo_{1}$ in it
    that satisfies all clauses in $\macdemands_{\macmdidx-1} \cup
    \maccomits_{\macmdidx-1}$. Add all such models
    $\macmodelo_{1}, \macmodelo_{2},\dots$ to our Kripke model
    $\macmodelo$ we are
    constructing and make the worlds $\macworldo_{1},
    \macworldo_{2},\dots$ successors of $\macworldo$. In $\macworldo$,
    set precisely those propositional variables to $\top$ that occur in
    $\mactlitsu_{\macmdidx}$. Let $\macclel$ be any clause in
    $\macclsu_{\macmdidx} \cup \maccomits_{\macmdidx-1}$. Now, we will
    prove by induction on modal depth of $\macclel$ that in $\macmodelo$,
    the world $\macworldo$ satisfies $\macclel$. If $\exists
    \maclitel\in \macclsc\setminus \mactlitsu : \macedgec ( \macclel,
    \maclitel)$ is true, then a propositional variable not in
    $\mactlitsu_{\macmdidx}$ occurs negatively in $\macclel$. Since this
    propositional variable is set to $\bot$ in $\macworldo$, $\macclel$
    is satisfied at $\macworldo$. If $\exists \maclitel\in
    \mactlitsu_{\macmdidx}: \macedge( \macclel, \maclitel)$ is true and
    $\maclitel$ is a propositional variable, then it is set to
    $\top$ in $\macworldo$ and occurs positively in $\macclel$. If $\exists \maclitel\in
    \mactlitsu_{\macmdidx}: \macedge( \macclel, \maclitel)$ is true and
    $\maclitel$ is of the form $\Box \maccl$, then the corresponding
    clause is in $\maccomits_{\macmdidx-1}$ and hence true in all
    successors of $\macworldo$ (including $\macworldo$ itself, by
    induction on modal depth of $\macclel$). If $\exists \maclitel\in
    \mactlitsu_{\macmdidx}: \macedge( \macclel, \maclitel)$ is true and
    $\maclitel$ is of the form $\Diamond \maccnf$, then we would have
    added a world to satisfy the corresponding $\maccnf$ formula.

    Now we will prove the other direction of the induction step.
    Suppose $\macclsu_{\macmdidx}$ is a subset of domain elements
    representing clauses occurring at level at most $\macmdidx$ and that
    there is a reflexive Kripke model $\macmodelo$ and a world
    $\macworldo$ such that $\macmodelo,\macworldo\models
    \maccnf(\macclsu_{\macmdidx})$. We will prove that
    $\macmdsat{\macmdidx}(\macclsu_{\macmdidx})$ is true in
    $\macstruc(\macmfo)$. To begin with, we will choose
    $\mactlitsu_{\macmdidx}$ to be the set of precisely those literals
    occurring at level $\macmdidx$ or below that are satisfied at
    $\macworldo$ and occur as subformulas of some clause in
    $\macclsu_{\macmdidx}$. This will ensure that $\forall \maclitel\in
    \mactlitsu_{\macmdidx} \quad \exists \macclel \in
    \macclsu_{\macmdidx}: \exists \macsovo: (\macreach( \macclel,
    \macsovo) \land \maclitel \in \macsovo)$ is true in $\macstruc (
    \macmfo)$. Let $\maccomits_{\macmdidx-1}=\{\macclel'\in
    \maccls\mid \exists \maclitel'\in
    \mactlitsu_{\macmdidx}\cap\macbox,\macedge(\maclitel',\macclel')\}$
    be the set of clauses such that the corresponding $\Box \maccl$ is
    in $\mactlitsu_{\macmdidx}$. The world $\macworldo$ satisfies all
    clauses in $\macclsu_{\macmdidx}$ and since $\macworldo$ is its own
    successor, it also satisfies all clauses in
    $\maccomits_{\macmdidx-1}$. Hence, if $\macclel$ is any clause in
    $\macclsu_{\macmdidx} \cup \maccomits_{\macmdidx-1}$, some literal
    occurring in $\macclel$ must be satisfied in $\macworldo$.
    Therefore, $\forall \macclel\in \macclsu_{\macmdidx} \cup
    \maccomits_{\macmdidx-1}:
     \left[\left(\exists \maclitel\in \mactlitsu_{\macmdidx} :
    \macedge(\macclel,\maclitel)\right)
     \lor \left(\exists \maclitel\in \macclsc
    \setminus \mactlitsu_{\macmdidx}:
    \macedgec(\macclel,\maclitel)\right)\right]$ is true in $\macstruc
    ( \macmfo)$.

    Let $\maclitel$ be any literal of the form $\Diamond \maccnf$ in
    $\mactlitsu_{\macmdidx}$ and let $\macdemands_{\macmdidx-1}=\{\macclel\in
    \maccls\mid
    \macedge(\maclitel,\macclel)\}$ be the set of clauses in the
    corresponding $\maccnf$ formula. Since $\macworldo$ satisfies
    $\maclitel$, there must be a successor $\macworldo'$ of $\macworldo$ that
    satisfies all clauses in $\macdemands_{\macmdidx-1}$ and also all
    clauses in $\maccomits_{\macmdidx-1}$ since $\macworldo'$ is a successor of
    $\macworldo$. Since all clauses in $\macdemands_{\macmdidx-1} \cup
    \maccomits_{\macmdidx-1}$ are at level at most $\macmdidx-1$ and
    $\macworldo'$ is a world in a reflexive Kripke model that satisfies
    all clauses in $\macdemands_{\macmdidx-1} \cup
    \maccomits_{\macmdidx-1}$, we can apply induction hypothesis to
    conclude that $\macmdsat{\macmdidx-1} ( \macdemands_{\macmdidx-1}
    \cup \maccomits_{\macmdidx-1})$ is true in $\macstruc( \macmfo)$.
    \qed
  \end{proof}

  \begin{theorem}
    \label{thm:modalRefSatFpt}
    Given a modal CNF formula $\macmfo$, there is a \macfpt\/ algorithm
    that checks if $\macmfo$ is satisfiable in reflexive models, with
    the treewidth of $\macstruc(\macmfo)$ and the modal depth of
    $\macmfo$ as parameters.
  \end{theorem}
  \begin{proof}
    Given $\macmfo$, $\macstruc(\macmfo)$ can be constructed in
    polynomial time. To check that all clauses of $\macmfo$ at level
    $\macmd(\macmfo)$ are satisfiable in some world $\macworldo$ of some
    reflexive Kripke model $\macmodelo$, we check whether the formula $\exists
    \macclsu_{\macmd(\macmfo)} \forall
    \macclel(\macclsu_{\macmd(\macmfo)}(\macclel)\Leftrightarrow
    (\maccls(\macclel)\land \maclev_{\macmd(\macmfo)}(\macclel)))\land
    \macmdsat{\macmd(\macmfo)}(\macclsu_{\macmd(\macmfo)})$ is true in
    $\macstruc(\macmfo)$. By \lemref{lem:mdSatMSO}, this is possible
    iff $\macmfo$ is satisfiable in a reflexive model. The length of the
    above formula is linear in $\macmd(\macmfo)$. An application of
    Courcelle's theorem will give us the \macfpt\/ algorithm.\qed
  \end{proof}
\end{minorext}

\section{Transitive models}
\label{sec:modSatTransFrames}
In transitive models, formulae with small modal depth can check
properties of all worlds reachable from a given world.  To formalize
this into a \macwone-hardness proof, we introduce the parameterized
Partitioned Weighted Satisfiability (\macppwcnfsat) problem. An
instance of \macppwcnfsat problem is a triple
$(\macpcnff,\macpart:\macpv\to[\macnumpart], \mactget:[k]\to
\macNat)$, where $\macpcnff$ is a propositional CNF formula,
$\macpart$ partitions the set of propositional variables into
$\macnumpart$ parts and we need to check if there is a satisfying
assignment that sets exactly $\mactget(\macpartidx)$ variables to
$\top$ in each part $\macpartidx$. Parameters are $\macnumpart$ and
pathwidth of the primal graph of $\macpcnff$ (one vertex for each
propositional variable, an edge between two variables iff they occur
together in a clause). The following lemma can be proved by a
\macfpt\/ reduction from the Number List Coloring Problem
\cite{FFLRSST07}.
\begin{lemma}%[*]
  \label{lem:pwsatHard}
  The \macppwcnfsat problem is \macwone-hard when parameterized by the
  number of parts $\macnumpart$ and the pathwidth of the primal graph.
\end{lemma}
\begin{details}
  \begin{proof}
    We will give a \macfpt\/ reduction from the Number List Coloring
    Problem (\macnlcp). An instance of \macnlcp\/ is a graph $\macgraph=
    (\macvertexs,\macedges)$, a set of colors $\maccols_{\macvertexo}$
    for each vertex $\macvertexo\in \macvertexs$ and a target function
    $\mactget:\cup_{\macvertexo\in \macvertexs}\maccols_{\macvertexo}\to
    \macNat$. We need to check if $\macgraph$ can be properly colored
    (every adjacent pair of vertices get different colors) such that
    every vertex $\macvertexo$ is colored from its set
    $\maccols_{\macvertexo}$ and there are exactly
    $\mactget(\maccolor)$ vertices colored with $\maccolor$ for
    every $\maccolor\in \cup_{\macvertexo\in
    \macvertexs}\maccols_{\macvertexo}$. In \cite{FFLRSST07}, it is
    proved that even for graphs of pathwidth $2$, \macnlcp\/ is
    \macwone-hard when parameterized by total number of colors in
    $\cup_{\macvertexo\in\macvertexs}\maccols_{\macvertexo}$.

    Given an instance of \macnlcp\/ with a graph of pathwidth $2$, we
    associate with it an instance of \macppwcnfsat with the set of
    propositional variables $\{\macpvo_{\macvertexo}^{\maccolor}\mid
    \macvertexo\in \macvertexs, \maccolor\in \maccols_{\macvertexo}\}$.
    Every color $\maccolor\in \cup_{\macvertexo\in
    \macvertexs}\maccols_{\macvertexo}$ is a
    partition of the set of propositional variables and contains the
    variables $\{\macpvo_{\macvertexo}^{\maccolor}\mid \maccolor\in
    \maccols_{\macvertexo}\}$. Target function is the same as target
    function of the \macnlcp\/ instance. The CNF formula is the
    conjunction of the following formulae:
    \begin{align*}
      \macatleast &\macdef \bigwedge_{\macvertexo\in \macvertexs}\left(
      \bigvee_{\maccolor\in
      \maccols_{\macvertexo}}\macpvo_{\macvertexo}^{\maccolor} \right)\\
      \macatmost &\macdef \bigwedge_{\macvertexo\in \macvertexs}
      \bigwedge_{\maccolor\ne\maccolor'\in\maccols_{\macvertexo}}\left(
      \lnot\macpvo_{\macvertexo}^{\maccolor}\lor \lnot
      \macpvo_{\macvertexo}^{\maccolor'} \right)\\
      \macproper&\macdef \bigwedge_{(\macvertexo,\macvertext)\in
      \macedges} \bigwedge_{\maccolor\in \maccols_{\macvertexo}\cap
      \maccols_{\macvertext}} \left( \lnot
      \macpvo_{\macvertexo}^{\maccolor} \lor
      \lnot\macpvo_{\macvertext}^{\maccolor} \right)
    \end{align*}
    Suppose the given \macnlcp\/ instance is a \macyes\/ instance. In the
    associated \macppwcnfsat instance, set
    $\macpvo_{\macvertexo}^{\maccolor}$ to $\top$ iff the vertex
    $\macvertexo$ receives color $\maccolor$ in the witnessing coloring.
    Since every vertex gets a color from its set, the formula
    $\macatleast$ above is satisfied. Since every vertex gets at most
    one color, the formula $\macatmost$ is satisfied. If
    $(\macvertexo,\macvertext)$ is any edge in the graph, then since
    $\macvertexo$ and $\macvertext$ get different colors in the
    witnessing coloring, the formula $\macproper$ above is also
    satisfied. Since target function of the \macppwcnfsat instance is
    same as the target function of the \macnlcp\/ instance, the target
    function of \macppwcnfsat is also satisfied.

    On the other hand, suppose that the instance of \macppwcnfsat is a
    \macyes\/ instance. Color a vertex $\macvertexo$ with the color
    $\maccolor$ iff the propositional variable
    $\macpvo_{\macvertexo}^{\maccolor}$ is set to $\top$ in the
    witnessing satisfying assignment. The formula $\macatleast$ ensures
    that every vertex gets at least one color from its set, while the
    formula $\macatmost$ ensures that every vertex gets at most one
    color. If $(\macvertexo,\macvertext)$ is an edge in $\macgraph$
    and $\maccolor$ is a common color between $\maccols_{\macvertexo}$
    and $\maccols_{\macvertext}$, then the formula $\macproper$ above
    ensures that at least one of the vertices
    $\macvertexo,\macvertext$ do not get the color $\maccolor$. Hence,
    the coloring given to the graph $\macgraph$ is proper. Again since
    target function of the \macppwcnfsat instance is same as the
    target function of the \macnlcp\/ instance, the target function of
    \macnlcp\/ is also satisfied.
    
    Now, it is left to prove that parameters of the \macppwcnfsat
    instance is bounded by some functions of the parameters of the
    \macnlcp\/ instance. First parameter of the \macppwcnfsat instance is
    the number of partitions, which is same as the total number of
    colors in the \macnlcp\/ instance (and later is a parameter of the
    \macnlcp\/ instance). Second parameter is the pathwidth of the primal
    graph of the CNF formula. Consider any path decomposition of width
    $2$ of the graph $\macgraph$ in the \macnlcp\/ instance. For every bag
    $\macbag$ and every vertex $\macvertexo$ in the bag, replace
    $\macvertexo$ by the set $\{\macpvo_{\macvertexo}^{\maccolor}\mid
    \maccolor\in \maccols_{\macvertexo}\}$. We claim that the resulting
    decomposition is a path decomposition of the primal graph of the CNF
    formula in the \macppwcnfsat instance. It is sufficient to prove
    that for every clause in the CNF formula, there is a bag containing
    all propositional variables occurring as literals in that clause.
    For any clause in the formula $\macatleast$ or $\macatmost$
    associated with a vertex $\macvertexo$, any bag that contained the
    vertex $\macvertexo$ before replacement will meet the above
    criteria. For a clause in the formula $\macproper$ associated with
    an edge $(\macvertexo,\macvertext)$, any bag that contained the
    vertices $\macvertexo$ and $\macvertext$ before replacement will
    suffice. In the new path decomposition, number of elements in any
    bag is at most $3$ times the total number of colors in the \macnlcp\/
    instance.  Hence, the pathwidth of the primal graph of the CNF formula
    in the \macppwcnfsat instance is also bounded by a function of the
    parameters of the \macnlcp\/ instance.\qed
  \end{proof}
\end{details}

\begin{theorem}
  \label{thm:modSatTransFramesHard}
  With treewidth and modal depth as parameters, modal satisfiability
  in transitive models is \macwone-hard.
\end{theorem}
The rest of this section is devoted to a proof of the above theorem,
which is by a \macfpt\/ reduction from \macppwcnfsat to satisfiability
of modal CNF formulae in transitive models. Given an instance
$(\macpcnff,\macpart:\macpv\to[\macnumpart],\mactget:[\macnumpart]\to
\macNat)$ of \macppwcnfsat problem with the pathwidth of the primal
graph of $\macpcnff$ being $\macpw$, we construct a modal CNF formula
$\macmfo_{\macpcnff}$ of modal depth $2$ in \macfpt\/ time such that
the pathwidth (and hence the treewidth) of
$\macstruc(\macmfo_{\macpcnff})$ is bounded by a function of $\macpw$
and $\macnumpart$ and \macppwcnfsat is a \macyes\/ instance iff
$\macmfo_{\macpcnff}$ is satisfiable in a transitive model. Suppose
the propositional variables used in $\macpcnff$ are
$\macpvo_{1},\macpvo_{2},\dots,\macpvo_{\macnumpv}$.  The idea is that
if $\macmfo_{\macpcnff}$ is satisfied at some world $\macworldo_{0}$
in some transitive model $\macmodelo$, then
$\macmodelo,\macworldo_{0}\models \macpcnff$. To check that the required
targets of the number of variables set to true in each partition are met,
$\macmfo_{\macpcnff}$ will force the existence of worlds
$\macworldo_{1},\macworldo_{2},\dots,\macworldo_{\macnumpv}$ arranged
as $\macworldo_{0}\macrelo\macworldo_{1}\macrelo\macworldo_{2}\macrelo
\cdots\macrelo\macworldo_{\macnumpv}$.  In the formula
$\macmfo_{\macpcnff}$, we will maintain a counter for each partition
of the propositional variables. At each world
$\macworldo_{\macpvidx}$, if $\macpvo_{\macpvidx}$ is true, we will
force the counter corresponding to $\macpart(\macpvo_{\macpvidx})$ to
increment. At the world $\macworldo_{\macnumpv}$, the counters will
have the number of variables set to $\top$ in each partition. We will
then verify in the formula $\macmfo_{\macpcnff}$ that these counts
meet the given target.  Such counting tricks have come under standard
usage in complexity theoretic arguments of modal logic. The challenge
here is to implement the counting in a modal formula of small
pathwidth.

In a \macppwcnfsat instance containing $\macnumpv$ propositional
variables and $\macnumpart$ partitions, we will denote the number of
variables in partition $\macpartidx$ by $\macnumpv[\macpartidx]$. We
first construct an optimal path decomposition of the primal graph of
$\macpcnff$ in \macfpt\/ time. We will name the variables occurring in
the first bag as $\macpvo_{1},\dots,\macpvo_{\macpvidx}$. We will name
the variables newly introduced in the second bag as
$\macpvo_{\macpvidx+1},\dots,\macpvo_{\macpvidx'}$ and so on. In the
rest of the construction, we will use this same ordering
$\macpvo_{1},\dots,\macpvo_{\macnumpv}$ of the propositional
variables. This will be important to maintain the pathwidth of the
resulting modal formula low. The modal CNF formula
$\macmfo_{\macpcnff}$ will use all the propositional variables
$\macpvo_{1},\dots,\macpvo_{\macnumpv}$ used by $\macpcnff$ and also
use the following additional variables:
\begin{itemize}
  \item
    $\macpartindo_{1},\dots,\macpartindo_{\macnumpart},\macpartindt_{1},
    \dots,\macpartindt_{\macnumpart}$: partition indicators.
  \item For each partition $\macpartidx$,
    $\macpartctro_{\macpartidx}^{0},\dots,
    \macpartctro_{\macpartidx}^{\macnumpv[\macpartidx]},
    \macpartctrt_{\macpartidx}^{0},\dots,
    \macpartctrt_{\macpartidx}^{\macnumpv[\macpartidx]}$: counters to
    count the number of variables set to $\top$ and $\bot$ in partition
    $\macpartidx$.
  \item $\macdpthctro_{0},\dots,\macdpthctro_{\macnumpv+1}$: depth
    indicators.
\end{itemize}
The modal CNF formula $\macmfo_{\macpcnff}$ is the conjunction of the
formulae described below. For clarity, we have used the shorthand notation
$\Rightarrow$ but they can be easily converted to CNF. Also for
notational convenience, we will use $\macpart(\macpvidx)$ instead of
$\macpart(\macpvo_{\macpvidx})$. $\macpv(\macpartidx)$ is the set of
variables among $\{\macpvo_{1},\dots,\macpvo_{\macnumpv}\}$ in
partition $\macpartidx$. The formula $\macdetermined$ ensures that all
successors of $\macworldo_{0}$ preserve the assignment of
$\macpvo_{1},\dots,\macpvo_{\macnumpv}$. The formula $\macdepth$
ensures that for all $\macpvidx$, $\macdpthctro_{\macpvidx}\land
\lnot\macdpthctro_{\macpvidx+1}$ holds in the world
$\macworldo_{\macpvidx}$.

In $\macworldo_{\macpvidx-1}$, if $\macpvo_{\macpvidx}$ is set to
$\top$, we want to indicate that in $\macworldo_{\macpvidx}$, the
counter for partition $\macpart(\macpvidx)$ should be incremented. We
will indicate this in the formula $\macsetcounter$ by setting the
variable $\macpartindo_{\macpart(\macpvidx)}$ to $\top$.  Similar
indication is done for the counter keeping track of variables set to
$\bot$ in partition $\macpartidx$.
\begin{align*}
  \macdetermined &\macdef
  \bigwedge_{\macpvidx=1}^{\macnumpv}\macpvo_{\macpvidx}\Rightarrow
  \Box\macpvo_{\macpvidx}\land
  \bigwedge_{\macpvidx=1}^{\macnumpv}\lnot\macpvo_{\macpvidx}\Rightarrow
  \Box\lnot\macpvo_{\macpvidx}\\
  \macdepth &\macdef \Diamond(\macdpthctro_{1} \land \lnot
  \macdpthctro_{2}) \land
  \bigwedge_{\macpvidx=1}^{\macnumpv-1}\Box\left[
  (\macdpthctro_{\macpvidx}\land \lnot \macdpthctro_{\macpvidx+1})\
  \Rightarrow  \Diamond(\macdpthctro_{\macpvidx+1}\land \lnot
  \macdpthctro_{\macpvidx+2})\right]\\
  \macsetcounter &\macdef (\macpvo_{1}\Rightarrow
  \macpartindo_{\macpart(1)}) \land
  (\lnot\macpvo_{1} \Rightarrow \macpartindt_{\macpart(1)})\\
   &\land  \bigwedge_{\macpvidx=2}^{\macnumpv}\Box\left\{
  [\macdpthctro_{\macpvidx-1}\land \lnot \macdpthctro_{\macpvidx}]
  \Rightarrow [(\macpvo_{\macpvidx}\Rightarrow
  \macpartindo_{\macpart(\macpvidx)}) \land
  (\lnot\macpvo_{\macpvidx}\Rightarrow
  \macpartindt_{\macpart(\macpvidx)})] \right\}\\
  \macinccounter &\macdef (\macpartindo_{\macpart(1)}\Rightarrow \Box
  \macpartctro_{\macpart(1)}^{1}) \land
  (\macpartindt_{\macpart(1)}\Rightarrow \Box
  \macpartctrt_{\macpart(1)}^{1})\\
  &\land \bigwedge_{\macpartidx=1}^{\macnumpart}
  \bigwedge_{\mactcidx=0}^{\macnumpv[\macpartidx]-1}
  \Box[\macpartindo_{\macpartidx}\Rightarrow
  (\macpartctro_{\macpartidx}^{\mactcidx}\Rightarrow \Box
  \macpartctro_{\macpartidx}^{\mactcidx+1})] \land
  \Box[\macpartindt_{\macpartidx}\Rightarrow
  (\macpartctrt_{\macpartidx}^{\mactcidx}\Rightarrow \Box
  \macpartctrt_{\macpartidx}^{\mactcidx+1})]\\
  \mactargetmet &\macdef \bigwedge_{\macpartidx=1}^{\macnumpart}\Box
  [\macdpthctro_{\macnumpv}\Rightarrow(\macpartctro_{\macpartidx}^{\mactget(\macpartidx)}
  \land \lnot\macpartctro_{\macpartidx}^{\mactget(\macpartidx)+1})]\\
  &\land \bigwedge_{\macpartidx=1}^{\macnumpart}\Box
  [\macdpthctro_{\macnumpv}\Rightarrow(\macpartctrt_{\macpartidx}^
  {\macnumpv[\macpartidx]-\mactget(\macpartidx)}
  \land \lnot\macpartctrt_{\macpartidx}^{\macnumpv[\macpartidx]-\mactget(\macpartidx)+1})]
\end{align*}
Variables $\macpartctro_{\macpartidx}^{0},\dots,
\macpartctro_{\macpartidx}^{\macnumpv[\macpartidx]}$ implement the
counter keeping track of variables set to $\top$ in partition
$\macpartidx$. If $\mactcidx$ variables in
$\macpv(\macpartidx)\cap\{\macpvo_{1},\dots,\macpvo_{\macpvidx}\}$ are
set to $\top$, then we want $\macpartctro_{\macpartidx}^{\mactcidx}$
to be set to $\top$ in $\macworldo_{\macpvidx}$. To maintain this, in
$\macworldo_{\macpvidx-1}$, if it is indicated that a counter is to be
incremented (by setting $\macpartindo_{\macpartidx}$ to $\top$), we
will force all successors of $\macworldo_{\macpvidx-1}$ to increment
the $\macpartctro_{\macpartidx}$ counter in the formula
$\macinccounter$.  Finally, we check that at $\macworldo_{\macnumpv}$,
all the targets are met in the formula $\mactargetmet$.

The modal CNF formula $\macmfo_{\macpcnff}$ we need is the conjunction
of $\macpcnff$, the formulae defined above and the miscellaneous
formulae below (which ensure that counters are initiated properly and
are monotonically non-decreasing).
\begin{align*}
  \macdetermined'&\macdef
  \bigwedge_{\macpartidx=1}^{\macnumpart}\macpartctro_{\macpartidx}^{0}
  \Rightarrow\Box \macpartctro_{\macpartidx}^{0} \land
  \bigwedge_{\macpartidx=1}^{\macnumpart} \macpartctrt_{\macpartidx}^{0}
  \Rightarrow\Box \macpartctrt_{\macpartidx}^{0}\\
  \maccountinit&\macdef \macdpthctro_{0} \land \lnot \macdpthctro_{1} \land
  \bigwedge_{\macpartidx=1}^{\macnumpart}(\lnot\macpartctro_{\macpartidx}^{1}
  \land \lnot\macpartctrt_{\macpartidx}^{1}\land
  \macpartctro_{\macpartidx}^{0}\land
  \macpartctrt_{\macpartidx}^{0})\\
  \macdepth'& \macdef \bigwedge_{\macpartidx=1}^{\macnumpart}
  \bigwedge_{\mactcidx=0}^{\macnumpv[\macpartidx]}\left[
  \Box(\macpartctro_{\macpartidx}^{\mactcidx}\Rightarrow
  \Box\macpartctro_{\macpartidx}^{\mactcidx}) \land \Box(
  \macpartctrt_{\macpartidx}^{\mactcidx}\Rightarrow
  \Box\macpartctrt_{\macpartidx}^{\mactcidx})  \right]\\
  \maccountmonotone &\macdef \bigwedge_{\macpvidx=1}^{\macnumpv}
  \Box(\macdpthctro_{\macpvidx}\Rightarrow \macdpthctro_{\macpvidx-1})
  \land \bigwedge_{\macpartidx=1}^{\macnumpart}
  \bigwedge_{\mactcidx=2}^{\macnumpv[\macpartidx]}\left[
  \Box(\macpartctro_{\macpartidx}^{\mactcidx}\Rightarrow
  \macpartctro_{\macpartidx}^{\mactcidx-1}) \land
  \Box(\macpartctrt_{\macpartidx}^{\mactcidx}\Rightarrow
  \macpartctrt_{\macpartidx}^{\mactcidx-1})
  \right]
\end{align*}
\begin{lemma}%[*]
  \label{lem:ppwSatImplMfSat}
  If a \macppwcnfsat instance is a \macyes\/ instance, then the modal
  formula constructed above is satisfied in a transitive Kripke
  model.
\end{lemma}
\begin{details}
  \begin{proof}
    We will construct a transitive Kripke model using the satisfying
    assignment $\macsatasgn$ that satisfies $\macpcnff$ while meeting
    the given target. The model $\macmodelo$ consists of worlds
    $\macworldo_{0},\macworldo_{1},\dots,\macworldo_{\macnumpv}$
    arranged as
    $\macworldo_{0}\macrelo\macworldo_{1}\macrelo\macworldo_{2}\macrelo
    \cdots\macrelo\macworldo_{\macnumpv}$. In all worlds,
    $\macpvo_{\macpvidx}$ is set to $\macsatasgn(\macpvo_{\macpvidx})$
    for all $\macpvidx$, thus ensuring that
    $\macmodelo,\macworldo_{0}\models \macpcnff\land \macdetermined$.
    In $\macworldo_{\macpvidx}$,
    $\{\macdpthctro_{0},\dots,\macdpthctro_{\macpvidx}\}$ are set to
    $\top$ and
    $\{\macdpthctro_{\macpvidx+1},\dots,\macdpthctro_{\macnumpv+1}\}$
    are set to $\bot$ for all $\macpvidx$ between $0$ and $\macnumpv$,
    thus ensuring that $\macmodelo,\macworldo_{0}\models
    \macdepth\land\macdpthctro_{0}\land\lnot\macdpthctro_{1}$, the
    last two clauses coming from the formula $\macdepth'$. It also
    ensures that $\macmodelo,\macworldo_{0}\models
    \bigwedge_{\macpvidx=1}^{\macnumpv}
    \Box(\macdpthctro_{\macpvidx}\Rightarrow
    \macdpthctro_{\macpvidx-1})$, which is part of
    $\maccountmonotone$.  We will set $\macpartctro_{\macpartidx}^{0}$
    and $\macpartctrt_{\macpartidx}^{0}$ to $\top$ in all worlds and
    $\macpartctro_{\macpartidx}^{1}$ and
    $\macpartctrt_{\macpartidx}^{1}$ to $\bot$ in $\macworldo_{0}$ for
    all partitions $\macpartidx$, thus ensuring
    $\macmodelo,\macworldo_{0}\models \maccountinit \land
    \macdetermined'$. At $\macworldo_{\macpvidx-1}$, we will set
    $\macpartindo_{\macpart(\macpvidx)}$ to $\macpvo_{\macpvidx}$'s
    value in the same world and $\macpartindt_{\macpart(\macpvidx)}$
    to $\lnot\macpvo_{\macpvidx}$'s value. This will ensure that
    $\macmodelo,\macworldo_{0}\models \macsetcounter$.
    
    At $\macworldo_{\macpvidx}$, for any partition $\macpartidx$, if
    $\mactcidx$ variables in $\macpv(\macpartidx)\cap
    \{\macpvo_{1},\dots,\macpvo_{\macpvidx}\}$ are set to $\top$, then
    we will set $\{\macpartctro_{\macpartidx}^{0},\dots,
    \macpartctro_{\macpartidx}^{\mactcidx}\}$ to $\top$ and
    $\{\macpartctro_{\macpartidx}^{\mactcidx+1},\dots,
    \macpartctro_{\macpartidx}^{\macnumpv[\macpartidx]}\}$ to $\bot$. If
    $\mactcidx'$ variables in $\macpv(\macpartidx)\cap
    \{\macpvo_{1},\dots,\macpvo_{\macpvidx}\}$ are set to
    $\bot$, we will set $\{\macpartctrt_{\macpartidx}^{0},\dots,
    \macpartctrt_{\macpartidx}^{\mactcidx'}\}$ to true and
    $\{\macpartctrt_{\macpartidx}^{\mactcidx'+1},\dots,
    \macpartctrt_{\macpartidx}^{\macnumpv[\macpartidx]}\}$ to $\bot$.
    For any $\macpartidx\ne \macpart(\macpvidx+1)$, we will set
    $\macpartindo_{\macpartidx}$ and $\macpartindt_{\macpartidx}$ to
    $\bot$ at $\macworldo_{\macpvidx}$. These will ensure that $\macmodelo,\macworldo_{0}\models
    \macinccounter\land \macdepth' \land \maccountmonotone$.

    Combined with the above settings of all propositional variables in
    $\macmodelo$, it is easy to check that the fact that $\macsatasgn$
    meets the target for each partition implies
    $\macmodelo,\macworldo_{0}\models \mactargetmet$.\qed
  \end{proof}
\end{details}
\begin{lemma}%[*]
  \label{lem:countingMfSatImplCountersSet}
  Suppose the modal CNF formula $\macmfo_{\macpcnff}$ constructed
  above is satisfied at some world $\macworldo_{0}$ of some transitive
  Kripke model $\macmodelo$. Then $\macmodelo$ contains distinct
  worlds $\macworldo_{1},\dots,\macworldo_{\macnumpv}$ such that for
  each $\macpvidx$ between $1$ and $\macnumpv$,
  $\macworldo_{\macpvidx}$ is a successor of
  $\macworldo_{\macpvidx-1}$. Moreover,
  $\{\macdpthctro_{0},\dots,\macdpthctro_{\macpvidx}\}$ are set to
  $\top$ and
  $\{\macdpthctro_{\macpvidx+1},\dots,\macdpthctro_{\macnumpv+1}\}$
  are set to $\bot$ in $\macworldo_{\macpvidx}$. For any partition
  $\macpartidx$, if $\mactcidx$ variables in $\macpv(\macpartidx)\cap
  \{\macpvo_{1},\dots,\macpvo_{\macpvidx}\}$ are set to $\top$ in
  $\macworldo_{0}$, then
  $\{\macpartctro_{\macpartidx}^{0},\dots,\macpartctro_{\macpartidx}^{\mactcidx}\}$
  are all set to $\top$ in $\macworldo_{\macpvidx}$. If $\mactcidx'$
  variables in $\macpv(\macpartidx)\cap
  \{\macpvo_{1},\dots,\macpvo_{\macpvidx}\}$ are set to $\bot$ in
  $\macworldo_{0}$, then
  $\{\macpartctrt_{\macpartidx}^{0},\dots,\macpartctrt_{\macpartidx}^{\mactcidx'}\}$
  are all set to $\top$ in $\macworldo_{\macpvidx}$.
\end{lemma}
\begin{details}
  \begin{proof}
    We will first prove the existence of worlds
    $\macworldo_{1},\dots,\macworldo_{\macpvidx}$ by induction on
    $\macpvidx$.

    \emph{Base case $\macpvidx=1$:} Since
    $\macmodelo,\macworldo_{0}\models \macdepth$, there must be a
    successor $\macworldo_{1}$ of $\macworldo_{0}$ that satisfies
    $\macdpthctro_{1}\land \lnot \macdpthctro_{2}$. Since
    $\macmodelo,\macworldo_{0}\models \maccountinit$,
    $\macworldo_{0}$ satisfies $\macdpthctro_{0}\land
    \lnot\macdpthctro_{1}$ and hence $\macworldo_{1}$ can not be same as
    $\macworldo_{0}$. Since $\macmodelo,\macworldo_{0}\models
    \Box(\macdpthctro_{3}\Rightarrow \macdpthctro_{2})$ (part of
    $\maccountmonotone$) and $\macworldo_{1}$ is a successor of
    $\macworldo_{0}$, we get $\macmodelo,\macworldo_{1}\models
    \macdpthctro_{3}\Rightarrow \macdpthctro_{2}$. Since
    $\macdpthctro_{2}$ is set to $\bot$ in $\macworldo_{1}$, this means
    that $\macdpthctro_{3}$ is also set to $\bot$ in
    $\macworldo_{1}$. Similar reasoning can be used to prove that all of
    $\{\macdpthctro_{2},\dots,\macdpthctro_{\macnumpv+1}\}$ are set to
    $\bot$ in $\macworldo_{1}$. The fact that
    $\macmodelo,\macworldo_{1}\models\macdpthctro_{1}\Rightarrow
    \macdpthctro_{0}$ means that $\macdpthctro_{0}$ is set to $\top$ in
    $\macworldo_{1}$ (since $\macdpthctro_{1}$ is set to $\top$ in
    $\macworldo_{1}$).

    \emph{Induction step:} Assume that worlds
    $\macworldo_{1},\dots,\macworldo_{\macpvidx}$ exist in $\macmodelo$
    with the stated properties. Hence, $\macworldo_{\macpvidx}$
    satisfies $\macdpthctro_{\macpvidx}\land \lnot
    \macdpthctro_{\macpvidx+1}$. Since $\macworldo_{0}$ satisfies
    $\macdepth$ and $\macworldo_{\macpvidx}$ is a successor of
    $\macworldo_{0}$ (by transitivity), there must be a successor
    $\macworldo_{\macpvidx+1}$ of $\macworldo_{\macpvidx}$ that
    satisfies $\macdpthctro_{\macpvidx+1} \land \lnot
    \macdpthctro_{\macpvidx+2}$. Since all worlds
    $\macworldo_{0},\dots,\macworldo_{\macpvidx}$ satisfy $\lnot
    \macdpthctro_{\macpvidx+1}$, $\macworldo_{\macpvidx+1}$ is distinct
    from all of them. The fact that $\macworldo_{\macpvidx+1}$
    satisfies $\macdpthctro_{\macpvidx'}\Rightarrow
    \macdpthctro_{\macpvidx'-1}$ for all $\macpvidx'$ (these formulae
    are part of $\maccountmonotone$ formula satisfied by
    $\macworldo_{0}$) can be used to show that all of
    $\macdpthctro_{0},\dots,\macdpthctro_{\macpvidx+1}$ are set to
    $\top$ in $\macworldo_{\macpvidx+1}$ and all of
    $\macdpthctro_{\macpvidx+2},\dots,\macdpthctro_{\macnumpv+1}$ are
    set to $\bot$ in $\macworldo_{\macpvidx+1}$.

    We will now prove the second claim of the lemma, which is about
    values of
    $\{\macpartctro_{\macpartidx}^{0},\dots,
    \macpartctro_{\macpartidx}^{\mactcidx}\}$ in
    $\macworldo_{\macpvidx}$. We will first prove that
    $\macpartctro_{\macpartidx}^{\mactcidx}$ is set to $\top$ by
    induction on $\macpvidx$.

    \emph{Base case $\macpvidx=1$:} If $\macpvo_{1}$ is not in part
    $\macpartidx$, there is nothing to prove
    ($\macpartctro_{\macpart(1)}^{0}$ is set to $\top$ in all worlds). If
    $\macpvo_{1}$ is in part $\macpartidx$ and $\macpvo_{1}$ is set to
    $\bot$, there is nothing to prove. If $\macpvo_{1}$ is in part
    $\macpartidx$ and $\macpvo_{1}$ is set to $\top$, then since
    $\macworldo_{0}$ satisfies $\macsetcounter$, we get
    $\macmodelo,\macworldo_{0}\models
    \macpvo_{1}\Rightarrow\macpartindo_{\macpart(1)}$. Since,
    $\macpvo_{1}$ is set to $\top$ and $\macpart(1)=\macpartidx$, we get
    that $\macpartindo_{\macpartidx}$ is set to $\top$ in
    $\macworldo_{0}$. Since $\macworldo_{0}$ satisfies $\macinccounter$,
    we get $\macmodelo,\macworldo_{0}\models
    \macpartindo_{\macpartidx}\Rightarrow
    \Box\macpartctro_{\macpartidx}^{1}$ and hence
    $\macmodelo,\macworldo_{0}\models \Box
    \macpartctro_{\macpartidx}^{1}$. Since $\macworldo_{1}$ is a
    successor of $\macworldo_{0}$, we conclude that in $\macworldo_{1}$,
    $\macpartctro_{\macpartidx}^{1}$ is set to $\top$.

    \emph{Induction step: Case 1:}  $\macpvo_{\macpvidx}$ is not in
    part $\macpartidx$ and none of the variables in $\macpv(\macpartidx)\cap
    \{\macpvo_{1},\dots,\macpvo_{\macpvidx}\}$ are set to to $\top$ in
    $\macworldo_{\macpvidx}$. In this case, there is nothing to prove.

    \emph{Case 2:} $\macpvo_{\macpvidx}$ is not in part $\macpartidx$
    and some $1\le \mactcidx<\macpvidx$ variables in
    $\macpv(\macpartidx)\cap \{\macpvo_{1},\dots,\macpvo_{\macpvidx}\}$
    are set to $\top$. By the induction hypothesis,
    $\macpartctro_{\macpartidx}^{\mactcidx}$ is set to $\top$ in
    $\macworldo_{\macpvidx-1}$. Now $\macmodelo,\macworldo_{0}\models
    \macdepth'$. Hence
    $\macmodelo,\macworldo_{0}\models\Box(\macpartctro_{\macpartidx}^{\mactcidx}
    \Rightarrow \Box\macpartctro_{\macpartidx}^{\mactcidx})$, and hence
    $\macmodelo,\macworldo_{\macpvidx-1}\models
    \macpartctro_{\macpartidx}^{\mactcidx}\Rightarrow \Box
    \macpartctro_{\macpartidx}^{\mactcidx}$ (since
    $\macworldo_{\macpvidx-1}$ is a successor of $\macworldo_{0}$), and
    hence $\macmodelo,\macworldo_{\macpvidx-1}\models \Box
    \macpartctro_{\macpartidx}^{\mactcidx}$ (since
    $\macpartctro_{\macpartidx}^{\mactcidx}$ is set to $\top$ in
    $\macworldo_{\macpvidx-1}$), and hence $\macmodelo,
    \macworldo_{\macpvidx}\models
    \macpartctro_{\macpartidx}^{\mactcidx}$ (since
    $\macworldo_{\macpvidx}$ is a successor of
    $\macworldo_{\macpvidx-1}$).

    \emph{Case 3:} $\macpvo_{\macpvidx}$ is in part $\macpartidx$ and
    $\macpvo_{\macpvidx}$ is set to $\bot$. If none of the variables in
    $\macpv(\macpartidx)\cap \{\macpvo_{1},\dots,\macpvo_{\macpvidx}\}$
    are set to $\top$, then the argument is similar to case 1. If some
    $1\le \mactcidx<\macpvidx$ variables in $\macpv(\macpartidx)\cap
    \{\macpvo_{1},\dots,\macpvo_{\macpvidx}\}$ are set to $\top$, then
    the argument is similar to case 2.

    \emph{Case 4:} $\macpvo_{\macpvidx}$ is in part $\macpartidx$ and
    $\macpvo_{\macpvidx}$ is set to $\top$. We know that
    $\macworldo_{\macpvidx-1}$ satisfies
    $\macdpthctro_{\macpvidx-1}\land \lnot \macdpthctro_{\macpvidx}$.
    Since $\macworldo_{0}$ satisfies $\macsetcounter$, we have
    $\macmodelo,\macworldo_{0}\models\Box\left\{
    [\macdpthctro_{\macpvidx-1}\land \lnot \macdpthctro_{\macpvidx}]
    \Rightarrow [\macpvo_{\macpvidx}\Rightarrow
    \macpartindo_{\macpart(\macpvidx)}] \right\}$, and hence
    $\macmodelo,\macworldo_{\macpvidx-1}\models
    [\macdpthctro_{\macpvidx-1}\land \lnot \macdpthctro_{\macpvidx}]
    \Rightarrow [\macpvo_{\macpvidx}\Rightarrow
    \macpartindo_{\macpart(\macpvidx)}]$ (since
    $\macworldo_{\macpvidx-1}$ is a successor of $\macworldo_{0}$), and
    hence $\macmodelo,\macworldo_{\macpvidx-1}\models
    \macpvo_{\macpvidx}\Rightarrow \macpartindo_{\macpartidx}$
    (since $\macmodelo,\macworldo_{\macpvidx-1}\models
    \macdpthctro_{\macpvidx-1}\land \lnot \macdpthctro_{\macpvidx}$),
    and hence $\macmodelo,\macworldo_{\macpvidx-1}\models
    \macpartindo_{\macpartidx}$ (since
    $\macmodelo,\macworldo_{\macpvidx-1}\models \macpvo_{\macpvidx}$).
    Since $\macworldo_{0}$ satisfies $\macinccounter$ and
    $\macworldo_{\macpvidx-1}$ is a successor of $\macworldo_{0}$, we
    get $\macmodelo,\macworldo_{\macpvidx-1}\models
    \macpartindo_{\macpartidx} \Rightarrow
    (\macpartctro_{\macpartidx}^{\mactcidx-1}\Rightarrow \Box
    \macpartctro_{\macpartidx}^{\mactcidx})$. We have already seen that
    $\macpartindo_{\macpartidx}$ is set to $\top$ in
    $\macworldo_{\macpvidx-1}$ and
    $\macpartctro_{\macpartidx}^{\mactcidx-1}$ is set to $\top$ in
    $\macworldo_{\macpvidx-1}$ by the induction hypothesis ($\mactcidx$ is
    at least $1$ since $\macpvo_{\macpvidx}$ is in part $\macpartidx$
    and is set to $\top$). Hence, we get
    $\macmodelo,\macworldo_{\macpvidx-1}\models
    \Box\macpartctro_{\macpartidx}^{\mactcidx}$. Since
    $\macworldo_{\macpvidx}$ is a successor of
    $\macworldo_{\macpvidx-1}$, we conclude that
    $\macpartctro_{\macpartidx}^{\mactcidx}$ is set to $\top$ in
    $\macworldo_{\macpvidx}$.
    
    Now, since $\macworldo_{0}$ satisfies
    $\Box(\macpartctro_{\macpartidx}^{\mactcidx}\Rightarrow
    \macpartctro_{\macpartidx}^{\mactcidx-1})$ (this is part of
    $\maccountmonotone$) and $\macworldo_{\macpvidx}$ is a successor of
    $\macworldo_{0}$, we get $\macmodelo,\macworldo_{\macpvidx}\models
    \macpartctro_{\macpartidx}^{\mactcidx}\Rightarrow
    \macpartctro_{\macpartidx}^{\mactcidx-1}$. Since
    $\macpartctro_{\macpartidx}^{\mactcidx}$ is set to $\top$ in
    $\macworldo_{\macpvidx}$, it follows that
    $\macpartctro_{\macpartidx}^{\mactcidx-1}$ is also set to $\top$ in
    $\macworldo_{\macpvidx}$. Similarly,
    $\macpartctro_{\macpartidx}^{0},\dots,\macpartctro_{\macpartidx}^{\mactcidx}$
    are all set to $\top$ in $\macworldo_{\macpvidx}$.
    
    The proof for values of $\{\macpartctrt_{\macpartidx}^{0},\dots,
    \macpartctrt_{\macpartidx}^{\mactcidx'}\}$ is symmetric to the proof
    of values of $\{\macpartctro_{\macpartidx}^{0},\dots,
    \macpartctro_{\macpartidx}^{\mactcidx}\}$.\qed
  \end{proof}
\end{details}
\begin{theorem}%[*]
  \label{thm:countingMfSatImplTargetMet}
  If $\macmfo_{\macpcnff}$ constructed above is satisfied in a
  transitive model, then the \macppwcnfsat instance is a \macyes\/
  instance.
\end{theorem}
\begin{details}
  \begin{proof}
    Suppose $\macmfo_{\macpcnff}$ is satisfied in some world
    $\macworldo_{0}$ of a transitive model. Since $\macpcnff$ is part of
    $\macmfo_{\macpcnff}$, the assignment to
    $\{\macpvo_{1},\dots,\macpvo_{\macnumpv}\}$ induced by
    $\macworldo_{0}$ satisfies $\macpcnff$. We claim that this
    assignment also meets the targets. If not, we will derive a
    contradiction. For some partition $\macpartidx$, suppose there are more
    than $\mactget(\macpartidx)$ variables set to $\top$. Then by
    \lemref{lem:countingMfSatImplCountersSet},
    $\macpartctro_{\macpartidx}^{\mactget(\macpartidx)+1}$ will be set
    to $\top$ in $\macworldo_{\macnumpv}$, contradicting the fact that
    $\macworldo_{0}$ satisfies $\mactargetmet$. For some partition
    $\macpartidx$, if there are less than $\mactget(\macpartidx)$ variables
    set to $\top$, then there will be more than
    $\macnumpv[\macpartidx]-\mactget(\macpartidx)$ variables set to
    $\bot$. By \lemref{lem:countingMfSatImplCountersSet},
    $\macpartctro_{\macpartidx}^{\macnumpv[\macpartidx]-\mactget(\macpartidx)+1}$
    will be set to $\top$ in $\macworldo_{\macnumpv}$, again
    contradicting the fact that $\macworldo_{0}$ satisfies
    $\mactargetmet$.\qed
  \end{proof}
\end{details}

Given an instance of \macppwcnfsat problem, the formula
$\macmfo_{\macpcnff}$ described above can be constructed in \macfpt\/
time. To complete the proof of \thmref{thm:modSatTransFramesHard}, we
will prove that the pathwidth of $\macmfo_{\macpcnff}$ is bounded by some
function of $\macnumpart$ and $\macpw$.  $\macmfo_{\macpcnff}$ has
been carefully constructed to keep pathwidth low.
\begin{lemma}%[*]
  \label{lem:pathwidthOfCountingMf}
  Pathwidth of $\macstruc(\macmfo_{\macpcnff})$ is at most
  $4\macpw+2\macnumpart+5$.
\end{lemma}
\begin{details}
  \begin{proof}
    Given an optimal path decomposition
    of the primal graph of $\macpcnff$, depth counters can be added to
    the bags without increasing their size much since the order of depth
    counters is same as the order of
    $\macpvo_{1},\dots,\macpvo_{\macnumpv}$. There are only $2\macnumpart$
    partition indicators
    $\macpartindo_{1},\dots,\macpartindo_{\macnumpart},\macpartindt_{1},
    \dots,\macpartindt_{\macnumpart}$, so they can also be added to the
    bags without increasing their size very much. However, the
    set of $2\macnumpv$ partition counters (of the form
    $\macpartctro_{\macpartidx}^{\mactcidx}$ or
    $\macpartctrt_{\macpartidx}^{\mactcidx}$) has to be added carefully
    to maintain the size of the bags. Formulas of $\macmfo_{\macpcnff}$
    have been carefully designed to enable this. The key observation is
    that the only ``link'' between $\macpvo_{1},\dots,\macpvo_{\macnumpv}$
    and partition counters are partition indicators and there are
    only $2\macnumpart$ of them. The following proof relies on
    this observation.

    Consider an optimal path decomposition of the primal graph of
    $\macpcnff$ with each bag containing at most $\macpw$ elements.
    Ensure that for all $\macpvidx$ with $1 \le
    \macpvidx<\macnumpv$, there is a bag containing both
    $\macpvo_{\macpvidx}$ and $\macpvo_{\macpvidx+1}$ or there is a bag
    with $\macpvo_{\macpvidx}$ such that the next bag contains
    $\macpvo_{\macpvidx+1}$ (call this the \emph{continuity} property).
    If this is not the case for some $\macpvidx$, consider the last bag
    $\macbag$ containing $\macpvo_{\macpvidx}$ and the first bag
    $\macbag'$ containing $\macpvo_{\macpvidx+1}$. No bag that is
    between $\macbag$ and $\macbag'$ will introduce any new variable (if
    it did, that new variable would have been $\macpvo_{\macpvidx+1}$
    according to our order). Hence, all the bags in between $\macbag$
    and $\macbag'$ are subsets of $\macbag$. Hence, they can all be
    removed and $\macbag'$ can become the bag immediately after
    $\macbag$. The resulting decomposition is still a path
    decomposition of the primal graph of $\macpcnff$ with each bag
    containing at most $\macpw$ elements. Moreover, the order of
    variables $\macpvo_{1},\dots,\macpvo_{\macnumpv}$ does not change
    due to the change we have made in the path decomposition. This new
    decomposition has a bag containing $\macpvo_{\macpvidx}$ such that
    the next bag contains $\macpvo_{\macpvidx+1}$. Now, we can repeat
    the above process until we get a path decomposition with the
    continuity property.
    
    For any $\macpvidx$ with $1\le \macpvidx\le \macnumpv$, let
    $\macbag_{\macpvidx}$ be a bag containing the propositional variable
    $\macpvo_{\macpvidx}$. We will expand this path decomposition by
    adding variables used in $\macmfo_{\macpcnff}$ such that for every
    $\maccl$ that appears in $\macmfo_{\macpcnff}$, there is a bag that
    contains all propositional variables appearing in that clause.
    Each of these expanded bags will have at most
    $4\macpw+2\macnumpart$ elements. We will then show how to expand
    this into a path decomposition of $\macstruc(\macmfo_{\macpcnff})$,
    by adding at most $6$ elements to each bag (creating duplicate
    copies of existing bags if required). This will prove that the pathwidth
    of $\macstruc(\macmfo_{\macpcnff})$ is at most
    $4\macpw+2\macnumpart+5$.

    First, in each bag $\macbag$ and each element
    $\macpvo_{\macpvidx}$ in it, add $\macdpthctro_{\macpvidx-1}, \macdpthctro_{\macpvidx}$
    and $\macdpthctro_{\macpvidx+1}$. Note that due to continuity
    property of the decomposition we started with, the expanded
    decomposition still retains the property that all bags containing
    an element forms a connected component, even after adding
    depth counters $\macdpthctro_{0},\dots,
    \macdpthctro_{\macnumpv+1}$. Next, add 
    $\macpartindo_{1},\dots,\macpartindo_{\macnumpart},
    \macpartindt_{1},\dots,\macpartindt_{\macnumpart}$ to all the
    bags.  We will refer to the bag containing
    $\macpvo_{\macpvidx},\macdpthctro_{\macpvidx-1},
    \macdpthctro_{\macpvidx}$ and $\macdpthctro_{\macpvidx+1}$ as
    $\macbag_{\macpvidx}$.  Now, we have a decomposition with each bag
    containing at most $4\macpw+2\macnumpart$ elements, and the last
    bag contains $\macdpthctro_{\macnumpv}$. To this
    bag, we will append $2\macnumpart$ paths serially. For $1 \le
    \macpartidx \le \macnumpart$,
    $(2\macpartidx-1)$\textsuperscript{th} path will be as follows:
    $\{\macdpthctro_{\macnumpv},\macpartindo_{1},\dots,\macpartindo_{\macnumpart},
    \macpartindt_{1},\dots,\macpartindt_{\macnumpart},\macpartctro_{\macpartidx}^{0},
    \macpartctro_{\macpartidx}^{1}\}-\{\macdpthctro_{\macnumpv},\macpartindo_{1},
    \dots,\macpartindo_{\macnumpart},
    \macpartindt_{1},\dots,\macpartindt_{\macnumpart},\macpartctro_{\macpartidx}^{1},
    \macpartctro_{\macpartidx}^{2}\}-\cdots-\{\macdpthctro_{\macnumpv},\macpartindo_{1},
    \dots,\macpartindo_{\macnumpart},
    \macpartindt_{1},\dots,\macpartindt_{\macnumpart},
    \macpartctro_{\macpartidx}^{\macnumpv[\macpartidx]-1},
    \macpartctro_{\macpartidx}^{\macnumpv[\macpartidx]}\}$. We will
    refer to these bags as
    $\macbag_{\macpartidx}^{1},\dots,\macbag_{\macpartidx}^{\macnumpv[\macpartidx]}$.
    $2\macpartidx$\textsuperscript{th} path is similar, with
    $\macpartctrt_{\macpartidx}$ variables replacing
    $\macpartctro_{\macpartidx}$ variables. We will refer to these bags in
    $2\macpartidx$\textsuperscript{th} path as
    $\macbag_{\macpartidx}^{1'},\dots,\macbag_{\macpartidx}^{\macnumpv[\macpartidx]'}$.
    Each of these new bags has at most $2\macnumpart+3$ elements, and
    the whole decomposition still retains the property that for any
    element, the set of bags containing that element forms a connected
    component.

    Now we will show how to expand the above decomposition into a path
    decomposition of $\macstruc(\macmfo_{\macpcnff})$. We have to add
    clauses and literals occurring in $\macmfo_{\macpcnff}$ and ensure
    that for any pair of elements
    $\macedge(\macdomel_{1},\macdomel_{2})$ or
    $\macedgec(\macdomel_{1},\macdomel_{2})$, there is a bag containing
    both $\macdomel_{1}$ and $\macdomel_{2}$. To achieve this, we may
    have to ``augment'' an existing bag with new elements. If
    $\macbag_{\macpvidx}$ is a bag in the path decomposition
    $\dots-\macbag_{\macpvidx}-\dots$, augmenting $\macbag_{\macpvidx}$
    with elements $\macdomel_{1}$ and $\macdomel_{2}$ means that we
    add another bag
    $\dots-\macbag'_{\macpvidx}-\macbag_{\macpvidx}-\dots$ with
    $\macbag'_{\macpvidx}$ containing all elements of
    $\macbag_{\macpvidx}$ and in addition containing
    $\macdomel_{1}$ and $\macdomel_{2}$. If we ensure that these new
    elements introduced during augmentation is never added to any other
    bag in the decomposition, augmentation will not violate the path
    decomposition's property that for any element, the set of bags
    containing that element forms a connected component. Now, we will go
    through each sub-formula of $\macmfo_{\macpcnff}$ and prove that all
    its clauses, literals and $\macedge$ pairs are already represented
    in the path decomposition we have constructed above or that the
    decomposition can be augmented to represent them.
    \begin{itemize}
      \item Clauses in $\macpcnff$: For each clause in $\macpcnff$,
	the propositional variables in that clause form a clique in the
	primal graph of $\macpcnff$. Hence, there is a bag $\macbag$ in
	the new decomposition that contains all propositional variables
	occurring in that clause. Augment $\macbag$ with a new domain
	element representing the clause.
      \item $\macdetermined$: Here, the clauses are of the form
	$\lnot\macpvo_{\macpvidx}\lor \Box\macpvo_{\macpvidx}$ and
	$\macpvo_{\macpvidx}\lor \Box\lnot\macpvo_{\macpvidx}$. Augment
	the bag $\macbag_{\macpvidx}$ containing
	$\macpvo_{\macpvidx}$ with $3$ domain elements, one for the
	clause $\lnot\macpvo_{\macpvidx}\lor \Box\macpvo_{\macpvidx}$
	itself, one for the literal $\Box\macpvo_{\macpvidx}$ and one
	for the clause in this literal that contains
	$\macpvo_{\macpvidx}$ as its only literal. Perform similar
	augmentation for the clause $\macpvo_{\macpvidx} \lor
	\Box\lnot\macpvo_{\macpvidx}$.
      \item $\macdepth$: For $\Diamond(\macdpthctro_{1}\land \lnot
	\macdpthctro_{2})$, augment the bag $\macbag_{1}$ containing
	$\macdpthctro_{1}$ and $\macdpthctro_{2}$ with
	$4$ domain elements representing literals and clauses of
	$\Diamond(\macdpthctro_{1}\land \lnot
	\macdpthctro_{2})$. Augment the bag $\macbag_{\macpvidx+1}$
	containing $\macdpthctro_{\macpvidx},
	\macdpthctro_{\macpvidx+1}$ and $\macdpthctro_{\macpvidx+2}$
	with $6$ elements representing literals and clauses of
	$\Box[\lnot\macdpthctro_{\macpvidx}\lor
	\macdpthctro_{\macpvidx+1}\lor
	\Diamond(\macdpthctro_{\macpvidx+1}\land
	\lnot\macdpthctro_{\macpvidx+2})]$.
      \item $\macsetcounter$: Augment the bag $\macbag_{1}$ containing
	$\macpvo_{1}$ and $\macpartindo_{\macpart(1)}$ with one
	element representing the clause $\lnot\macpvo_{1}\lor
	\macpartindo_{\macpart(1)}$. Do a similar augmentation for the
	clause $\macpvo_{1}\lor \macpartindt_{\macpart(1)}$.
	$\Box(\macpvo\land\macpvt)$ is equivalent to $\Box\macpvo \land
	\Box\macpvt$. Hence, the latter part of $\macsetcounter$ can be
	split into clauses $\Box(\lnot\macdpthctro_{\macpvidx-1}\lor
	\macdpthctro_{\macpvidx}\lor \lnot\macpvo_{\macpvidx}\lor
	\macpartindo_{\macpart(\macpvidx)})$ and $\Box(\lnot\macdpthctro_{\macpvidx-1}\lor
	\macdpthctro_{\macpvidx}\lor \macpvo_{\macpvidx}\lor
	\macpartindt_{\macpart(\macpvidx)})$. Augment the bag
	$\macbag_{\macpvidx}$ containing
	$\macdpthctro_{\macpvidx-1}, \macdpthctro_{\macpvidx},
	\macpvo_{\macpvidx}, \macpartindo_{\macpart(\macpvidx)}$ and
	$\macpartindt_{\macpart(\macpvidx)}$ with $6$ elements
	representing clauses and literals of these two clauses.
      \item $\macinccounter$: Augment the bag
	$\macbag_{\macpart(1)}^{1}$ containing
	$\macpartindo_{\macpart(1)}$ and
	$\macpartctro_{\macpart(1)}^{1}$ with $3$ elements representing
	clauses and literals of $(\lnot\macpartindo_{\macpart(1)}\lor
	\Box\macpartctro_{\macpart(1)}^{1})$. Similarly augment the bag
	$\macbag_{\macpart(1)}^{1'}$ for
	$(\lnot\macpartindt_{\macpart(1)}\lor
	\Box\macpartctrt_{\macpart(1)}^{1})$. Augment the bag
	$\macbag_{\macpartidx}^{\mactcidx+1}$ containing
	$\macpartindo_{\macpartidx},
	\macpartctro_{\macpartidx}^{\mactcidx}$ and
	$\macpartctro_{\macpartidx}^{\mactcidx+1}$ with $6$ elements
	representing literals and clauses of $\Box(\lnot
	\macpartindo_{\macpartidx} \lor \lnot
	\macpartctro_{\macpartidx}^{\mactcidx} \lor
	\Box\macpartctro_{\macpartidx}^{\mactcidx+1})$. Similarly
	augment $\macbag_{\macpartidx}^{\mactcidx+1'}$ for
	$\Box(\lnot \macpartindt_{\macpartidx} \lor \lnot
	\macpartctrt_{\macpartidx}^{\mactcidx}\lor
	\Box\macpartctrt_{\macpartidx}^{\mactcidx+1})$.
      \item $\mactargetmet$: Augment the bag
	$\macbag_{\macpartidx}^{\mactget(\macpartidx)+1}$ containing
	$\macdpthctro_{\macnumpv},
	\macpartctro_{\macpartidx}^{\mactget(\macpartidx)}$ and
	$\macpartctro_{\macpartidx}^{\mactget(\macpartidx)+1}$ with
	$6$ elements for the literals and clauses in $\Box(\lnot
	\macdpthctro_{\macnumpv}\lor
	\macpartctro_{\macpartidx}^{\mactget(\macpartidx)})$ and
	$\Box(\lnot \macdpthctro_{\macnumpv}\lor \lnot
	\macpartctro_{\macpartidx}^{\mactget(\macpartidx)+1})$.
	Similarly augment
	$\macbag_{\macpartidx}^{\macnumpv[\macpartidx]-\mactget(\macpartidx)+1'}$ for
	$\Box(\lnot \macdpthctro_{\macnumpv}\lor
	\macpartctrt_{\macpartidx}^{\macnumpv[\macpartidx]-\mactget(\macpartidx)})$
	and $\Box(\lnot \macdpthctro_{\macnumpv}\lor
	\lnot\macpartctrt_{\macpartidx}^{\macnumpv[\macpartidx]-\mactget(\macpartidx)+1})$
      \item $\macdetermined'$: Augment the bag
	$\macbag_{\macpartidx}^{1}$ containing
	$\macpartctro_{\macpartidx}^{0}$ with $3$ elements representing
	literals and clauses of $\lnot
	\macpartctro_{\macpartidx}^{0}\lor
	\Box\macpartctro_{\macpartidx}^{0}$. Similarly augment
	$\macbag_{\macpartidx}^{1'}$ for $\lnot
	\macpartctrt_{\macpartidx}^{0} \lor
	\Box\macpartctrt_{\macpartidx}^{0}$.
      \item $\maccountinit$: Augment the bag $\macbag_{1}$ containing
	$\macdpthctro_{0}$ and $\macdpthctro_{1}$ with $2$ elements
	representing the clauses in $\macdpthctro_{0}\land \lnot
	\macdpthctro_{1}$. Augment the bag
	$\macbag_{\macpartidx}^{1}$ containing
	$\macpartctro_{\macpartidx}^{0}$ and
	$\macpartctro_{\macpartidx}^{1}$ with $2$ elements representing
	the clauses in $\lnot \macpartctro_{\macpartidx}^{1}\land
	\macpartctro_{\macpartidx}^{0}$. Similarly augment
	$\macbag_{\macpartidx}^{1'}$ for $\lnot
	\macpartctrt_{\macpartidx}^{1} \land
	\macpartctrt_{\macpartidx}^{0}$.
      \item $\macdepth'$: Augment the bag
	$\macbag_{\macpartidx}^{\mactcidx}$ containing
	$\macpartctro_{\macpartidx}^{\mactcidx}$ with $6$ elements
	representing literals and clauses of $\Box(\lnot
	\macpartctro_{\macpartidx}^{\mactcidx} \lor
	\Box\macpartctro_{\macpartidx}^{\mactcidx})$. Similarly augment
	$\macbag_{\macpartidx}^{\mactcidx'}$ for $\Box(\lnot
	\macpartctrt_{\macpartidx}^{\mactcidx} \lor
	\Box\macpartctrt_{\macpartidx}^{\mactcidx})$.
      \item $\maccountmonotone$: Augment the bag $\macbag_{\macpvidx}$
	containing $\macdpthctro_{\macpvidx}$ and
	$\macdpthctro_{\macpvidx-1}$ with $3$ elements representing
	literals and clauses of $\Box(\lnot \macdpthctro_{\macpvidx}
	\lor \macdpthctro_{\macpvidx-1})$. Augment the bag
	$\macbag_{\macpartidx}^{\mactcidx}$ containing
	$\macpartctro_{\macpartidx}^{\mactcidx}$ and
	$\macpartctro_{\macpartidx}^{\mactcidx-1}$ with $3$ elements
	representing literals and clauses of $\Box(\lnot
	\macpartctro_{\macpartidx}^{\mactcidx}\lor
	\macpartctro_{\macpartidx}^{\mactcidx-1})$. Similarly augment
	$\macbag_{\macpartidx}^{\mactcidx'}$ for $\Box(\lnot
	\macpartctrt_{\macpartidx}^{\mactcidx}\lor
	\macpartctrt_{\macpartidx}^{\mactcidx-1})$.\qed
    \end{itemize}
  \end{proof}
\end{details}
In the absence of transitivity, the above reduction would require a
formula of modal depth that depends on $\macnumpv$ (and hence it would
no longer be a \macfpt\/ reduction). The above hardness proof will
however go through for any class of transitive frames that has paths
of unbounded length of the form
$\macworldo_{1}\macrelo\macworldo_{2}\macrelo\cdots\macrelo
\macworldo_{\macnumpv}$ without any reverse paths\footnote{The author
acknowledges an anonymous referee for pointing this out.}. See
\cite{BTC09} for some context on such classes of transitive frames of
unbounded depth.
\begin{details}
  Readers familiar with frames with no branching to the right
  (axiom .3) may infer that the above hardness proof will also go
  through for K4.3 frames.
\end{details}

\section{Conclusions and Future Work}
By expressing satisfiability of modal formulae as a MSO property,
we obtained a \macfpt\/ algorithm for modal satisfiability in general
models with treewidth and modal depth as parameters. Due to the
dependence of the constructed MSO sentence on modal depth, the
\macfpt\/ algorithm obtained in \secref{sec:modSatGenFrames} has a
running time with a tower of $2$'s whose height is
$\macoh(\macmd(\macmfo))$.  Unless, \macp=\macnp, such dependence on
modal depth cannot be avoided due to the following observation. In
\cite[Lemma 1]{ALM09}, it is shown how to encode an arbitrary
propositional CNF formula into an equivalent modal formula (the
propositional formula is satisfiable iff the modal formula is
satisfiable in a general model). This modal formula has some very low
modal depth $\macheight$ such that any function growing slower than a
tower of $2$'s of height $\macheight-5$ is a polynomial in the size of
the propositional formula. The treewidth of this modal formula can be
verified to be a constant. This also proves that unless \macp=\macnp,
modal satisfiability in general models is not \macfpt\/ when treewidth
is the only parameter.

We can work out a composition algorithm \cite{BDFH09}, and hence
conclude that with treewidth and modal depth as parameters, there is
no polynomial kernel for modal satisfiability in general models.

One direction for future research is towards meta classification as
done in \cite{HS08}, instead of the case by case analysis of this
work. We can also consider variations in treewidth, such as having
different domain elements representing same propositional variable at
different levels in $\macstruc(\macmfo)$. Other variations are modal
circuits instead of modal formulae and generalizations of primal/dual
graphs instead of incidence graphs.

{\bf Acknowledgements.} The author wishes to thank Kamal Lodaya,
Geevarghese Philip and Saket Saurabh for helpful discussions, pointers
to related work and feedback on the draft. The author also thanks
anonymous referees of a previous version of this paper for catching
some subtle errors and suggesting extensions.
\bibliographystyle{plain}
\bibliography{references}

\end{document}